\documentclass[a4paper]{article}
\usepackage{amsmath}
\usepackage{amssymb}
\usepackage{dsfont}
\usepackage{mathrsfs}
\usepackage{mathabx}
\usepackage{stmaryrd}
\usepackage[all]{xy}
\usepackage{tikz}
\usepackage{graphicx}
\usetikzlibrary{graphs}

\title{A Set-Theoretic Framework for Parallel Graph Rewriting}

\author{Thierry  Boy de la Tour  \and  Rachid Echahed}
\date{Univ. Grenoble Alpes, CNRS, Grenoble INP, LIG \\ 38000 Grenoble,
  France \\
{\small \texttt{thierry.boy-de-la-tour@imag.fr} and \texttt{rachid.echahed@imag.fr} }}

\usepackage[amsmath,thmmarks]{ntheorem}
\theoremstyle{plain}
\theorembodyfont{\itshape}
\theoremheaderfont{\normalfont\bfseries}
\theoremseparator{}
\newtheorem{theorem}{Theorem} [section]
\newtheorem{lemma}[theorem]{Lemma}
\newtheorem{corollary}[theorem]{Corollary}

\theorembodyfont{\upshape}
\theoremheaderfont{\normalfont\bfseries}
\theoremsymbol{\ensuremath{_\Box}}
\newtheorem{definition}[theorem]{Definition}
\newtheorem{example}[theorem]{Example}

\theoremstyle{nonumberplain}
\theoremheaderfont{\itshape}
\theorembodyfont{\normalfont}
\theoremsymbol{\ensuremath{_{\blacksquare}}}
\newtheorem{proof}{Proof}

\newcommand{\set}[1]{\{#1\}}
\newcommand{\setof}[2]{\{#1\mid #2\}}

\newcommand{\card}[1]{|#1|}
\newcommand{\tuple}[1]{\langle #1 \rangle}

\newcommand{\ensvide}{\varnothing}
\newcommand{\Sym}[1]{\mathrm{Sym}(#1)}
\newcommand{\eqclass}[2]{{#1}[#2]}
\newcommand{\quotient}[2]{{#1}/{#2}}
\newcommand{\Part}[1]{\mathscr{P}(#1)}
\newcommand{\defeq}{\stackrel{\mathrm{\scriptscriptstyle def}}{=}}

\newcommand{\restrfc}[3]{#1\vert_{#2}^{#3}}
\newcommand{\restrf}[2]{\restrfc{#1}{#2}{}}
\newcommand{\invf}[1]{#1^{-1}}
\newcommand{\Id}[1]{\mathrm{Id}_{#1}}
\newcommand{\meetf}{\curlywedge}
\newcommand{\joinf}{\curlyvee}
\newcommand{\bigjoinf}{\bigcurlyvee}

\newcommand{\Alg}{\mathcal{A}}
\newcommand{\algebrf}[1]{\mathscr{A}_{#1}}
\newcommand{\labelf}[1]{\mathring{#1}}
\newcommand{\carrier}[1]{\lfloor{#1}\rfloor}

\newcommand{\nodes}[1]{\dot{#1}}
\newcommand{\arrows}[1]{\vec{#1}}
\newcommand{\leftf}[1]{\acute{#1}}
\newcommand{\rightf}[1]{\grave{#1}}
\newcommand{\subg}{\mathrel{<}}
\newcommand{\subalg}{\mathrel{\lhd}}
\newcommand{\restrg}[2]{{#1}|{#2}}
\newcommand{\sgminus}{\setminus}
\newcommand{\delgraph}[4]{#1\setminus[#2,#3,#4]}
\newcommand{\sgen}[2]{\lceil{#2}\rceil_{#1}}
\newcommand{\gcap}{\sqcap}
\newcommand{\gcup}{\sqcup}
\newcommand{\Gcup}{\bigsqcup}
\newcommand{\isog}{\mathrel{\equiv}}
\newcommand{\G}{\mathcal{G}}

\newcommand{\am}{\alpha}
\newcommand{\bm}{\beta}
\newcommand{\mm}{\mu}
\newcommand{\nm}{\nu}
\newcommand{\Idm}[1]{{1}_{#1}}
\newcommand{\Aut}[2]{\mathrm{Aut}_{#1}(#2)}
\newcommand{\liftm}[1]{\mbox{$#1\!\!\uparrow$}}
\newcommand{\eqaut}{\approx}

\newcommand{\Termsig}[1]{\mathscr{T}(\Sigma,#1)}
\newcommand{\Vars}{\mathscr{V}}
\newcommand{\Var}[1]{\mathrm{Var}(#1)}

\newcommand{\Matches}[2]{\mathscr{M}(#1,#2)}
\newcommand{\R}{\mathcal{R}}

\newcommand{\M}{\mathcal{M}}
\newcommand{\rulem}[1]{\mathrm{r}_{#1}}
\newcommand{\Lg}[1]{\mathrm{L}_{#1}}
\newcommand{\Kg}[1]{\mathrm{K}_{#1}}
\newcommand{\Rg}[1]{\mathrm{R}_{#1}}
\newcommand{\nodesLg}[1]{\nodes{\mathrm{L}}_{#1}}
\newcommand{\nodesKg}[1]{\nodes{\mathrm{K}}_{#1}}
\newcommand{\nodesRg}[1]{\nodes{\mathrm{R}}_{#1}}
\newcommand{\arrowsLg}[1]{\arrows{\mathrm{L}}_{#1}}
\newcommand{\arrowsKg}[1]{\arrows{\mathrm{K}}_{#1}}
\newcommand{\arrowsRg}[1]{\arrows{\mathrm{R}}_{#1}}

\newcommand{\leftfRg}[1]{\leftf{\mathrm{R}}_{#1}}

\newcommand{\labelfLg}[1]{\labelf{\mathrm{L}}_{#1}}
\newcommand{\labelfKg}[1]{\labelf{\mathrm{K}}_{#1}}
\newcommand{\labelfRg}[1]{\labelf{\mathrm{R}}_{#1}}
\newcommand{\RImg}[2]{\liftm{#1}_{\!#2}}

\newcommand{\PRgraphl}[2]{{#1}^{<}_{#2}}
\newcommand{\PRgraphr}[2]{{#1}^{>}_{#2}}
\newcommand{\fullPRl}[1]{\mathrel{\rightrightarrows^{<}_{#1}}}
\newcommand{\fullPRr}[1]{\mathrel{\rightrightarrows^{>}_{#1}}}
\newcommand{\PRautl}[1]{\mathrel{\rightrightharpoons^{<}_{#1}}}
\newcommand{\PRautr}[1]{\mathrel{\rightrightharpoons^{>}_{#1}}}

\newcommand{\inN}[1]{\iota_{#1}}
\newcommand{\nodesinN}[1]{\nodes{\iota}_{#1}}
\newcommand{\arrowsinN}[1]{\arrows{\iota}_{#1}}
\newcommand{\labelfinN}[1]{\labelf{\iota}_{#1}}
\newcommand{\sigmA}[1]{\sigma_{#1}}
\newcommand{\nodesigmA}[1]{\nodes{\sigma}_{#1}}
\newcommand{\arrowsigmA}[1]{\arrows{\sigma}_{#1}}
\newcommand{\labelfsigmA}[1]{\labelf{\sigma}_{#1}}
\newcommand{\taU}[1]{\tau_{#1}}
\newcommand{\labelftaU}[1]{\labelf{\tau}_{#1}}
\newcommand{\rhO}[1]{\rho_{#1}}
\newcommand{\nodesrhO}[1]{\nodes{\rho}_{#1}}
\newcommand{\arrowsrhO}[1]{\arrows{\rho}_{#1}}
\newcommand{\labelfrhO}[1]{\labelf{\rho}_{#1}}
\begin{document}
\maketitle

\begin{abstract}
  We tackle the problem of attributed graph transformations and
  propose a new algorithmic approach for defining parallel graph
  transformations allowing overlaps. We start by introducing some
  abstract operations over graph structures. Then, we define the
  notion of rewrite rules as three inclusions of the form
  $L \supseteq K \supseteq M \subseteq R$. We provide six conditions
  that parallel graph rewrite relations should ideally satisfy, which
  lead us to define two distinct full parallel graph rewrite
  relations.  A central notion of regularity of matchings is proved to
  be equivalent to these six conditions, and to the equality of
  these two relations.  Furthermore, we take advantage of the
  symmetries that may occur in $L$, $K$, $M$ and $R$ and define
  another pair of rewrite relations that factor out possibly many
  equivalent matchings up to their common symmetries. These
  definitions and the corresponding proofs combine operations on
  graphs with group-theoretic notions, thus illustrating the relevance
  of our framework.
\end{abstract}

\section{Introduction}\label{sec-intro}

Graph structures are widely used in many areas in computer science and
well  beyond   (e.g.,  Biology,  Chemistry,  Physics).   Their  visual
appearance  as well  as their  expressiveness give  them an  important
place in the modeling of complex systems.

The study of graph transformations turns out to be more difficult than
other structures such as strings \cite{BO93} or terms \cite{BN98}. One
of  the  main  issues  encountered in  graph  transformations  is  the
replacement action. Roughly speaking, let $G[L]$  be a graph and $L$ a
subgraph of it. Replacing $L$ by another  graph $R$ is a bit tricky to
express in a  rigorous manner, because of the  possible links between
$L$  and its  context  $G[]$, on  the  one hand,  and  of the  desired
embedding of $R$ into $G[]$, on the other hand.

Several approaches to graph transformations  have been proposed in the
literature.  There are  two  main  streams of  research  known as  the
algebraic  approaches (see,  e.g;, \cite{handbook1,EhrigEPT06})  where
transformations  are  defined  using notions  borrowed  from  category
theory         and          the         algorithmic         approaches
(e.g. \cite{EngelfrietR97,Echahed08b}) where graph transformations are
defined by means of the involved algorithms.

In this paper, we propose a new algorithmic approach and define a
set-theoretic framework which we designed to cope easily with (true)
parallel attributed graph transformations and where classical
sequential graph transformations can be obtained as a particular
case. We propose to consider rules defined as three graph inclusions
of the form $L \supseteq K \supseteq M \subseteq R$ such that $M$ is
the common subgraph of $L$ and $R$. Roughly speaking, transforming a
graph $G$ into a graph $H$ via such a rule aims at finding a match
$\mu$ of the left-hand side $L$ in $G$, so that one can write $G$ as
$G[\mu(L)]$ ; then $H$ is obtained from $G[\mu(L)]$, first, by
deleting elements in $\mu(L \setminus K)$\footnote{Here, we use the
  operation $\setminus$ over graphs in an informal manner.}  and,
second, by adding new elements in $R \setminus M$. As $M \subseteq K$
and $M$ is the common subgraph of $L$ and $R$, $R \setminus M$ and
$R \setminus K$ are the same and thus one may have the impression, at
first sight, that subgraph $M$ is useless. This impression might be
reinforced when comparing with the classical Double-Pushout
\cite{EhrigPS73} approach where rules can be represented as two
inclusions $L' \supseteq K' \subseteq R'$ with the same intuitive
semantics. Actually, if we get rid of subgraph $M$, our rules could be
translated to a Double-Pushout rule of the form
$L \supseteq K \subseteq (R \cup K)$\footnote{In Double-Pushout
  approach, there is a morphism between $K$ and the right-hand side.}.

The main motivation behind the use of the fourth graph $M$ in a
rule lies in the operational semantics of parallel graph
transformations. To keep it simple, one rule applied to a graph can
either delete, keep or add items to the considered graph. When at
least two rules are applied at the same time on a graph, they can
either behave independently if applied at disjoint parts of a graph or
the two rules have to agree on the parts of the graph to be
kept/deleted if their respective left-hand sides overlap (have a
non-empty intersection). So, graphs such as $M$ play a key role,
within the considered rules, when defining the operational semantics
of parallel graph transformations. Intuitively, when a rule is
involved in a parallel transformation of a graph, the subgraph matched
by $L \setminus K$ should be removed, the subgraph $M$ should be kept,
and $R \setminus M$ should be added. However, the part of the graph
matched by $K \setminus M$ could either be deleted by other rules or
kept unchanged. Notice that there is no equivalent counterpart of
graphs $K \setminus M$ in the algebraic approaches \cite{EhrigEPT06}.
 
To define rigorously the considered rewrite systems and the underlying
rewrite relations, we propose, in this paper, a
set-theoretic framework consisting of operations over graphs which
allowed us to express (parallel) rewrite steps in an elegant way. To
our knowledge these operations are different from those existing in
the literature such as \cite{BVG87,DershowitzJ2018,Echahed08b}. In
addition to sequential rewriting, we propose two possible definitions
of parallel graph transformations which differ by the order in which
the operations over graphs are applied. Furthermore, we provide
necessary and sufficient conditions under which the two proposed
parallel rewrite relations coincide. Finally, we propose to use group theory to
characterize enhanced rewrite relations which consider only matchings up to
automorphisms of the four graphs defining each rule. 

Parallel graph transformations did not receive so far as much attention as
the sequential case. One of the most investigated issue is 
the condition so-called \emph{parallel independence} which ensures that two
rules with matches in a same graph G, are independent, i.e. they can
be applied in any order (or even in parallel) yielding the same
result, see e.g. \cite{CorradiniMREHL97,Lowe93-Parallel}. Our approach
to the study of parallel graph transformations is different here since
the parallel transformations we consider cannot always be simulated by
sequential ones, as in the case of cellular automata or more generally
in substitution systems \cite{Wolfram02}, and thus the involved rules
cannot be ordered to form a sequential derivation equivalent to
parallel transformations.  

In \cite[chapter~14]{Plasmeijer93},
parallel graph transformations have been studied in order to improve
the operational semantics of the functional programming language CLEAN
\cite{cleanURL}. In that contribution, the authors do not deal with
true parallelism but rather have an interleaving semantics. This
particularly entails that their parallel rewriting steps can be
simulated by sequential ones. This is also the case for other
frameworks where massive parallel graph transformations is defined so
that it can be simulated by sequential rewriting e.g.,
\cite{KreowskiKL18,KreowskiK11}.

In \cite{KniemeyerBHK07}, a framework based on the algebraic
Single-Pushout approach has been proposed and where parallel
transformations consider only matchings provided by a control flow
mapping. The users can solve the possible conflicts between the rules
by providing the right control flow.  More recently, a parallel graph
rewriting has been defined in \cite{EchahedM17} for a special kind of
graphs called port-graphs. Unfortunately, such graphs are not closed
under parallel graph rewriting, in the sense that a port-graph can be
rewritten in a structure which is not a port-graph. In addition,
conditions for avoiding conflicts in parallel transformations have
been defined over the considered rewrite rules, which limits the class
of considered systems, meanwhile we provide in this paper more
abstract and more general conditions over matchings that ensure the
correctness of parallel graph transformations.

The paper is organized as follows. Basic definitions and notations are
gathered in the next section. In Section~\ref{sec-join}, we stress on
the notion of \emph{joinable graphs} and provide a definition of
automorphisms of sub-graphs. Section~\ref{sec-rules} is dedicated to
rules and matches while Section~\ref{sec-full-rew} specifies the
intended requirements of parallel graph transformations and defines
two parallel rewrite relations. In Section~\ref{sec-regularity}, we
propose a property of \emph{regularity} of set of matches which turns
out to be a necessary and sufficient condition on matches to make
equal the two different proposed parallel rewrite
relations. Section~\ref{sec-Aut} is dedicated to a particular parallel
rewrite relation where parallel matches are considered up to
automorphisms. This rewrite relation is based on the notion of automorphism
groups of the considered rules.  An additional example is provided in
Section~\ref{sec-example}. Finally, concluding remarks are given in
Section~\ref{sec-conclusion}.

\section{Preliminaries}\label{sec-defs}

We use the standard set theoretic notion of a total function as a functional binary relation, but we also need to be strict on the notion of domain and codomain of functions. Since the domain can be read out of the binary relation, we define a \emph{function $f$ from $D$ to $C$} as a tuple $f=\tuple{R,C}$ where $R\subseteq D\times C$ and $\forall x\in D, \exists!y\in C$ s.t. $\tuple{x,y}\in R$. Both the \emph{domain} $D$ and the \emph{codomain} $C$ can be extracted from the function $f$: $D$ is the first projection $\setof{x}{\exists y\in C,\, \tuple{x,y}\in R}$ of $R$ (the second projection of $R$ is the \emph{image} of $f$). We will use the standard abuse of notation by denoting $f$ the canonical extension of $f$ from $\Part{D}$ to $\Part{C}$, i.e., for any $A\subseteq D$ we write $f(A)$ for $\setof{f(x)}{x\in A}$. For any $y\in C$, we write $\invf{f}(y)$ for $\setof{x\in D}{f(x)=y}$; this is only used when $f$ is not assumed to be bijective, otherwise it could be confused with the image of $y$ by the inverse function $\invf{f}=\tuple{\invf{R},D}$. For any $B\subseteq C$, we write $\invf{f}(B)$ for $\setof{x\in D}{f(x)\in B}$. The \emph{restriction} of $f=\tuple{R,C}$ to the sets $D',C'$ is $\restrfc{f}{D'}{C'} \defeq \tuple{R\cap (D'\times C'),C'}$, which is a function from $D\cap D'$ to $C'$ provided that $f(D\cap D')\subseteq C'$. We write $\restrf{f}{D'}$ for $\restrfc{f}{D'}{C}$, which is always a function.

We consider a fixed many sorted signature $\Sigma$ (see
e.g. \cite{EhrigM85}). A \emph{graph} $G$ is a tuple
$\tuple{V,A,\Alg,s,t,l}$ where $V,A$ are sets and $\Alg$ is a
$\Sigma$-algebra whose elements are respectively called
\emph{vertices}, \emph{arrows} and \emph{labels}, $s,t$ are functions
from $A$ to $V$ and $l$ is a function from $V\cup A$ to $\Part{\Alg}$,
i.e., we label vertices and arrows with sets of labels. We assume that
$V$, $A$ and the carrier set of $\Alg$, denoted $\carrier{\Alg}$, are
mutually disjoint\footnote{This condition is not strictly necessary,
  but it allows, e.g., the use of $V\uplus A$ rather than
  $V+A = V\times\set{1}\uplus A\times\set{2}$ and hence of simpler set
  theoretic notations.}. The \emph{carrier} of $G$ is the set
$\carrier{G}\defeq V\uplus A\uplus \carrier{\Alg}$. If $V=\ensvide$
then $G$ is said to be \emph{empty}. An arrow $f\in A$ is said to go
\emph{from} $s(f)$ \emph{to} $t(f)$, and these two vertices are
\emph{adjacent} to $f$ (or $f$ to the vertices). When we speak of a
graph $G$ without specifying its components, these will be referred to
as in
$G=\tuple{\nodes{G},\arrows{G},\algebrf{G},\leftf{G},\rightf{G}, \labelf{G}}$.

A graph $H$ is a \emph{subgraph} of $G$, written $H\subg G$, if $\nodes{H}\subseteq\nodes{G}$, $\arrows{H}\subseteq\arrows{G}$, $\leftf{H}=\restrfc{\leftf{G}}{\arrows{H}}{\nodes{H}}$ and $\rightf{H}=\restrfc{\rightf{G}}{\arrows{H}}{\nodes{H}}$. $H$ is a \emph{$\Sigma$-subgraph} of $G$, written $H\subalg G$, if $H\subg G$, $\algebrf{H} = \algebrf{G}$ and $\forall x\in \nodes{H}\uplus\arrows{H}, \labelf{H}(x)\subseteq \labelf{G}(x)$. The relations $\subg$ and $\subalg$ are respectively a preorder and an order on graphs.

There are several ways to obtain a $\Sigma$-subgraph from a graph $G$, one is by selecting arrows in $A\subseteq \arrows{G}$: let $\restrg{G}{A}\defeq\tuple{\nodes{G},A,\algebrf{G},\restrf{\leftf{G}}{A},\restrf{\rightf{G}}{A}, \restrf{\labelf{G}}{\nodes{G}\cup A}}$ be the \emph{restriction} of $G$ to $A$; we have $\restrg{G}{A}\subalg G$. Thus we can also remove the arrows of $A$ by defining $G \sgminus A \defeq \restrg{G}{(\arrows{G}\setminus A)}$. We can also obtain $\Sigma$-subgraphs by selecting vertices, but then we also need to select their adjacent arrows: if $V\subseteq\nodes{G}$, let $A=\invf{\leftf{G}}(V)\cap \invf{\rightf{G}}(V)$, then $\sgen{G}{V}\defeq\tuple{V, A, \algebrf{G}, \restrfc{\leftf{G}}{A}{V},\restrfc{\rightf{G}}{A}{V}, \restrf{\labelf{G}}{V\cup A}}$ is the \emph{$\Sigma$-subgraph of $G$ generated} by $V$: $\restrfc{\leftf{G}}{A}{V}$ and $\restrfc{\rightf{G}}{A}{V}$ are obviously functions, hence $\sgen{G}{V}\subalg G$. We can then remove the vertices of $V$ by $G\sgminus V \defeq\sgen{G}{\nodes{G}\setminus V}$. 

The task of removing labels from a graph $G$ may differ according to
vertices or arrows, i.e., it is specified by a function from
$\nodes{G}\cup\arrows{G}$ to $\Part{\algebrf{G}}$, which we call a
\emph{labelling function of} $G$. Given two such functions $l$ and
$l'$ we define $l\setminus l'$ (resp. $l\cap l'$, $l\cup l'$) as the
labelling function of $G$ that maps any $x$ to $l(x)\setminus l'(x)$
(resp. $l(x)\cap l'(x)$, $l(x)\cup l'(x)$). If $l$ is a labelling
function of a graph $H\subg G$, we extend it to the labelling function
$l'$ of $G$ identical to $l$ on $\nodes{H}\cup\arrows{H}$ and that
maps any other $x$ to $\emptyset$; by abuse of notation $l'$ will be
denoted by $l$. Then we easily define the graph $G\sgminus l \defeq
\tuple{\nodes{G},\arrows{G},\algebrf{G},
  \leftf{G},\rightf{G}, \labelf{G}\setminus l}$. We thus see that
using sets of labels allows to remove, and later to add labels just
as we do arrows, which will be very convenient for defining parallel rewrite relations.

Finally, we can remove both $V$, $A$ and $l$ from $G$ by defining $\delgraph{G}{V}{A}{l} \defeq ((G\sgminus l)\sgminus A)\sgminus V$. This order comes from the requirement that $V$ and $A$ should be included in the carrier set of the graph from which they are removed, which is used in the following lemma.

\begin{lemma}\label{lm-subremove}
  If $H\subalg G$ then $H\subalg \delgraph{G}{V}{A}{l}$ iff $\nodes{H}\cap V = \arrows{H}\cap A = \ensvide$ and $\labelf{H}(x)\cap l(x) = \ensvide$ for all $x\in \nodes{H}\cup\arrows{H}$.
\end{lemma}
\begin{proof}
  Since $H\subalg G$, then $H\subalg G\sgminus l$ iff $\forall x\in \nodes{H}\cup \arrows{H}$, $\labelf{H}(x)\subseteq \labelf{G}(x)\setminus l(x)$ iff $\forall x\in \nodes{H}\cup \arrows{H}$, $\labelf{H}(x)\cap l(x) = \ensvide$ (since $\labelf{H}(x)\subseteq \labelf{G}(x)$).

Since the set of arrows of $G\sgminus l$ is $\arrows{G}$, then $H\subalg (G\sgminus l)\sgminus A$ iff $H\subalg G\sgminus l$ and $\arrows{H}\subseteq \arrows{G}\setminus A$ iff $H\subalg G\sgminus l$ and $\arrows{H}\cap A = \ensvide$.

Since the set of vertices of $(G\sgminus l)\sgminus A$ is $\nodes{G}$ then $H\subalg \delgraph{G}{V}{A}{l}$ iff $H\subalg (G\sgminus l)\sgminus A$ and $\nodes{H}\subseteq \nodes{G}\setminus V$ iff $H\subalg (G\sgminus l)\sgminus A$ and $\nodes{H}\cap V=\ensvide$. Hence the result.
\end{proof}

For any graph $G$ and any bijection $\am$ from $\carrier{G}$ to some set $S$, we define the restrictions $\nodes{\am} \defeq \restrfc{\am}{\nodes{G}}{\am(\nodes{G})}$, $\arrows{\am} \defeq \restrfc{\am}{\arrows{G}}{\am(\arrows{G})}$, $\labelf{\am} \defeq \restrfc{\am}{\carrier{\algebrf{G}}}{\am(\carrier{\algebrf{G}})}$ and then the structure 
\[\am(G) \defeq\tuple{\nodes{\am}(\nodes{G}),\ \arrows{\am}(\arrows{G}),\ \labelf{\am}(\algebrf{G}),\ \nodes{\am}\circ\leftf{G}\circ\invf{\arrows{\am}},\  \nodes{\am}\circ\rightf{G}\circ \invf{\arrows{\am}},\ \labelf{\am}\circ\labelf{G}\circ \invf{(\restrf{\am}{\nodes{G}\cup \arrows{G}})}}\]
where $\labelf{\am}(\algebrf{G})$ is the isomorphic image of $\algebrf{G}$ by the bijection $\labelf{\am}$, see \cite{EhrigM85}. Obviously, $\carrier{\am(G)} = S$ and $\am(G)$ is a graph; we say that $\am$ is an \emph{isomorphism from $G$ to $\am(G)$}. It is common knowledge that any mathematical structure (including $\Sigma$-algebras) can be translated in this way through a bijective function, which then becomes an isomorphism.

Given any two graphs $H$ and $G$, a function $\am$ from $\carrier{H}$ to $\carrier{G}$ is a \emph{morphism from $H$ to $G$} if
\begin{itemize}
\item $\restrfc{\am}{\nodes{H}}{\nodes{G}}$ is a function, denoted $\nodes{\am}$,
\item $\restrfc{\am}{\arrows{H}}{\arrows{G}}$ is a function, denoted $\arrows{\am}$,
\item $\restrfc{\am}{\carrier{\algebrf{H}}}{\carrier{\algebrf{G}}}$ is a $\Sigma$-homomorphism from $\algebrf{H}$ to $\algebrf{G}$, denoted $\labelf{\am}$,
\item  $\leftf{G}\circ\arrows{\am} = \nodes{\am}\circ\leftf{H}$,
\item $\rightf{G}\circ\arrows{\am} = \nodes{\am}\circ\rightf{H}$,
\item $\forall x\in\nodes{H}\cup\arrows{H}$, $\labelf{\am}\circ\labelf{H}(x)\subseteq \labelf{G}\circ \am(x)$.
\end{itemize}
If $\am$ is bijective then the functions $\nodes{\am}$, $\arrows{\am}$ and $\labelf{\am}$ correspond with the previous definition, which justifies that we use the same notations. The \emph{image} of a $\Sigma$-subgraph $F\subalg H$ by $\am$ is \[\am(F) \defeq \tuple{\nodes{\am}(\nodes{F}), \arrows{\am}(\arrows{F}), \algebrf{G}, \restrfc{\leftf{G}}{\arrows{\am}(\arrows{F})}{\nodes{\am}(\nodes{F})},\restrfc{\leftf{G}}{\arrows{\am}(\arrows{F})}{\nodes{\am}(\nodes{F})}, l}\]\[\text{where }\forall y\in \nodes{\am}(\nodes{F})\cup \arrows{\am}(\arrows{F}),\ l(y)\ =\ \bigcup_{x\in \invf{\am}(y)}\labelf{\am}\circ \labelf{F}(x),\] 
obviously $\am(F)\subalg \am(H)\subalg G$. Again, if $\am$ is bijective it is easy to see that $\am(H)$ corresponds with the previous definition, which of course does not mean that $\am(H)=G$, since $G$ may contain more labels than $\am(H)$. A function $\am$ from $\carrier{H}$ to $\carrier{G}$ is an isomorphism from $H$ to $G$ if and only if $\am$ is a bijective morphism from $H$ to $G$ \emph{and} $\invf{\am}$ is a morphism from $G$ to $H$, hence if and only if $\am$ is a bijective morphism \emph{and} $\labelf{\am}\circ\labelf{H} = \labelf{G}\circ \restrf{\am}{\nodes{H}\cup \arrows{H}}$. 

If $\nodes{\am}$ and $\arrows{\am}$ are injective then $\am$ is called a \emph{matching of $H$ in $G$}. An isomorphism from $G$ to $G$ is called an \emph{automorphism} of $G$. The identity function $\Id{\carrier{G}}$, denoted $\Idm{G}$, is always an automorphism of $G$. If $\am$ is a morphism from $H$ to $G$ and $\bm$ a morphism from $G$ to a graph $F$ then $\bm\circ\am$ is a morphism from $H$ to $F$. If $\am$ and $\bm$ are both matchings, isomorphisms or automorphisms then so is $\bm\circ\am$. The composition operator $\circ$ is associative, we have $\Idm{G}\circ\am = \am\circ\Idm{H} = \am$ and, if $\am$ is an isomorphism then $\am\circ\invf{\am}=\Idm{G}$ and $\invf{\am}\circ\am = \Idm{H}$. Hence the set of automorphisms of $G$ is a group denoted $\Aut{}{G}$, and the existence of isomorphisms between graphs is an equivalence relation denoted $\isog$.

Note that $\Aut{}{G}$ contains permutations of the set $\carrier{G}$, and its product is the standard composition of permutations, hence it is a \emph{permutation group} on $\carrier{G}$. More precisely, it is a subgroup of the \emph{symmetric group} $\Sym{\carrier{G}}$ of all permutations on $\carrier{G}$, defined by $\Aut{}{G}=\setof{\am\in\Sym{\carrier{G}}}{\am(G)=G}$. As an example, we apply the permutation $\am = (x\ f\ g\ y)$, in cycle notation, to the following graph:
\[\am\Bigg( \xymatrix{x \ar@/^/[r]^f \ar@/_/[r]_g & y} \Bigg) \ =\  \xymatrix{f \ar@/^/[r]^g \ar@/_/[r]_y & x}
\]
(labels are hidden), the result is obviously not the input graph. This
is a very abstract view of things (it defines an \emph{operation} of
$\Sym{\carrier{G}}$ on the class of graphs built on $\carrier{G}$, and
$\Aut{}{G}$ as the \emph{stabilizer} of $G$ in $\Sym{\carrier{G}}$ by
this operation, see e.g. \cite[Chapter II, Section 1]{Hoffmann82} for a
nice introduction), and in practice we do not need the whole group $\Sym{\carrier{G}}$. Obviously when $\am(G)=G$ we have $\nodes{\am}(\nodes{G})=\nodes{G}$, hence $\nodes{\am}\in\Sym{\nodes{G}}$, similarly $\arrows{\am}\in\Sym{\arrows{G}}$ and $\labelf{\am}$ is a $\Sigma$-automorphism of $\algebrf{G}$, the set of which should be denoted $\Aut{}{\algebrf{G}}$. 
Following the previous example, if we take $\nodes{\am}=()$, $\arrows{\am}=(f\ g)$ and $\labelf{\am}=()$ (assuming that $f$ and $g$ have the same label), then
\[\am\Bigg( \xymatrix{x \ar@/^/[r]^f \ar@/_/[r]_g & y} \Bigg) \ =\  \xymatrix{x \ar@/^/[r]^g \ar@/_/[r]_f & y}
\]
which is exactly the input graph, hence $(f\ g)$ is an automorphism of this graph.

\section{Joinable Graphs}\label{sec-join}

In order to define parallel rewriting relations on graphs, it is convenient to join possibly many different graphs that have a common part, i.e., that are joinable. As a matter of fact, this notion also allows a simple definition of graph rewriting rules, and is crucial in defining the automorphism groups of these rules. We start with a simpler notion of joinable functions. Basic properties are given without proofs.

\begin{definition}[joinable functions]\label{def-joinablef}
Given two functions $f=\tuple{R,C}$ and $g=\tuple{R',C'}$ of domain $D$, $D'$ respectively, we define the \emph{meet} of $f$ and $g$ as $f \meetf g \defeq \tuple{R\cap R',C\cap C'}$, which is a function whose domain is a subset of $D\cap D'$ (it is the set of all $x\in D\cap D'$ such that $f(x)=g(x)$). If this domain is exactly $D\cap D'$ then the \emph{join} $f \joinf g \defeq \tuple{R\cup R',C\cup C'}$ is also a function (from $D\cup D'$ to $C\cup C'$), and we say that $f$ and $g$ are \emph{joinable}. 

Similarly, if $(f_i)_{i\in I}$ is an $I$-indexed family of pairwise joinable functions, where $f_i=\tuple{R_i,C_i}$ has domain $D_i$, then $\bigjoinf_{i\in I}f_i \defeq \tuple{\ \bigcup_{i\in I}R_i,\ \bigcup_{i\in I} C_i\ }$ is a function of domain $\bigcup_{i\in I}D_i$.

If $S$ and $T$ are sets of functions, let $S\joinf T \defeq \setof{f\joinf g}{f\in S,\, g\in T}$ and $S\circ T \defeq \setof{f\circ g}{f\in S,\, g\in T}$, provided these operations can be applied. If $f$ is a function, let $f\circ T \defeq \set{f}\circ T$.
\end{definition}

In particular, functions with disjoint domains are joinable (e.g. $\nodes{\am}$ and $\arrows{\am}$), and every function is joinable with itself: $f\joinf f= f\meetf f = f$.
More generally, any two restrictions $\restrf{f}{A}$ and $\restrf{f}{B}$ of the same function $f$ are joinable and $\restrf{f}{A}\joinf \restrf{f}{B} = \restrf{f}{A\cup B}$ (and of course $\restrf{f}{A}\meetf \restrf{f}{B} = \restrf{f}{A\cap B}$). Conversely, if $f$ and $g$ are joinable then each is a restriction of $f\joinf g$.

\begin{definition}[joinable graphs]\label{def-joinableg} 
 Two graphs $H$ and $G$ are \emph{joinable} if $\algebrf{H}=\algebrf{G}$, $\nodes{H}\cap\arrows{G} = \arrows{H}\cap\nodes{G} = \ensvide$, and the functions $\leftf{H}$ and $\leftf{G}$ (and similarly $\rightf{H}$ and $\rightf{G}$) are joinable.
We can then define the graphs 
\begin{eqnarray*}
  H\gcap G &\defeq& \tuple{\ \nodes{H}\cap\nodes{G},\ \arrows{H}\cap\arrows{G},\  \algebrf{H},\ \leftf{H}\meetf\leftf{G},\ \rightf{H}\meetf\rightf{G},\   \restrf{(\labelf{H}\cap\labelf{G})}{(\nodes{H}\cap\nodes{G})\cup(\arrows{H}\cap\arrows{G})}\ },\\
  H\gcup G &\defeq& \tuple{\ \nodes{H}\cup\nodes{G},\  \arrows{H}\cup\arrows{G},\ \algebrf{H},\ \leftf{H}\joinf\leftf{G},\ \rightf{H}\joinf\rightf{G},\ \labelf{H}\cup\labelf{G}\ }.
\end{eqnarray*}
Similarly, if $(G_i)_{i\in I}$ is an $I$-indexed family of graphs (where $I\neq\ensvide$) that are pairwise joinable, hence have the same algebra $\Alg$ of labels, then let \[\Gcup_{i\in I}G_i \ \defeq\ \tuple{\ \bigcup_{i\in I}\nodes{G_i},\ \bigcup_{i\in I}\arrows{G_i},\ \Alg,\ \bigjoinf_{i\in I}\leftf{G_i},\ \bigjoinf_{i\in I}\rightf{G_i},\ \bigcup_{i\in I}\labelf{G_i}\ }.\]
\end{definition}

It is easy to see that these structures are graphs: the sets of vertices and arrows are disjoint and the adjacency functions have the correct domains and codomains. Note that if $H$ and $G$ are joinable then any subgraphs of $H$ and $G$ are joinable, and $H\gcap G = G\gcap H\subalg H \subalg H\gcup G = G\gcup H$. Similarly, if the $G_i$'s are pairwise joinable then $\forall j\in I$, $G_j\subalg \Gcup_{i\in I}G_i$. We also see that $H\subalg G$ entails that $H$ and $G$ are joinable and then $H\gcap G = H$ and $H\gcup G = G$.

We now extend the notion of automorphism groups of graphs to their $\Sigma$-subgraphs.
\begin{definition}[groups $\Aut{G}{H_1,\ldots,H_n}$ and $\restrf{\mathcal{S}}{H}$]\label{def-Aut-in}
  For any $n\geq 1$ and any graphs $H, H_1,\ldots,H_n \subalg G$, let 
  \[\Aut{G}{H}  \defeq  \setof{\am\in\Sym{\nodes{G}}\joinf\Sym{\arrows{G}}\joinf\Aut{}{\algebrf{G}}}{\am(H)=H}\]
  \[\text{and }\Aut{G}{H_1,\ldots,H_n}  \defeq  \bigcap_{i=1}^n\Aut{G}{H_i}.\]
For any $\am\in\Aut{G}{H}$, we write $\restrf{\am}{H}$ for $\restrfc{\am}{\carrier{H}}{\carrier{H}}$, and for any subgroup $\mathcal{S}$ of $\Aut{G}{H}$, let $\restrf{\mathcal{S}}{H} = \setof{\restrf{\am}{H}}{\am\in \mathcal{S}}$; this is a subgroup of $\Aut{}{H}$.
\end{definition}

It is obvious that $\Aut{G}{G}=\Aut{}{G}$. We see that $\Aut{G}{H}$ is
a permutation group on $\carrier{G}$, but only the graph structure of
$H$ is involved in the constraint $\am(H)=H$, not the structure of
$G$. 

\begin{example}
  Take for instance
\[H\ =\ \xymatrix{x \ar@/^/[r]^f \ar@/_/[r]_g & y}\ \text{ and }\ G\ =\ \xymatrix{x \ar@/^/[r]^f \ar@/_/[r]_g & y \ar@/^/[r]^h & z \ar@/^/[l]^k }\]
where labels are omitted. We have \[\Aut{}{H} = \set{\Idm{H},\, (x)(y)(f\ g)}\text{ and }\Aut{}{G} = \set{\Idm{G},\, (x)(y)(z)(f\ g)(h)(k)}\] (we write fixpoints in order to make the domains explicit). However, in $\Aut{G}{H}$ the permutations of objects that do not belong to $H$ are free, hence
\begin{eqnarray*}
  \Aut{G}{H} &=& \set{\Idm{G},\, (x)(y)(z)(f\ g)(h)(k),\, (x)(y)(z)(f)(g)(h\ k),\\
  && \hspace*{15em}(x)(y)(z)(f\ g)(h\ k)}\\
  &=& \Aut{}{H}\joinf \set{(z)(h)(k),\,(z)(h\ k)}\\
  &=& \Aut{}{H}\joinf \set{(z)} \joinf \set{(h)(k),\ (h\ k)}\\
  &=& \Aut{}{H}\joinf \Sym{\set{z}}\joinf \Sym{\set{h,\,k}}.
\end{eqnarray*}
\end{example}

Since $\algebrf{H}=\algebrf{G}$, it is easy to see that \[\Aut{G}{H}=\Aut{}{H}\joinf\Sym{\nodes{G}\setminus\nodes{H}}\joinf\Sym{\arrows{G}\setminus\arrows{H}}\]
always holds and hence that $\restrf{\Aut{G}{H}}{H}=\Aut{}{H}$. This means that, compared to the elements of $\Aut{}{H}$ which are only permutations of $\carrier{H}$, the elements of $\Aut{G}{H}$ are \emph{all} possible extensions of the elements of $\Aut{}{H}$ to permutations of $\carrier{G}$. This contrasts with the standard group theoretic way of extending permutations $\pi$ of a set $A\subseteq B$ to permutations of $B$ by adding only fixpoints, thus assuming that $\forall x\in B\setminus A$, $\pi(x)=x$. This will allow us to conveniently intersect the automorphism groups of joinable graphs, see Section \ref{sec-Aut}.

One important problem with the notion of joinable graphs is that the union of graphs, as defined above, is not an invariant construction, i.e., joining isomorphic images of joinable graphs may not be possible, or may not yield an isomorphic image of the union of the original graphs. To ensure that this is the case, we need some extra conditions.
\begin{lemma}\label{lm-isojoin}
  Let $\am$ (resp. $\bm$) be an isomorphism from graph $H$ to $H'$ (resp. $G$ to $G'$), if $H$ and $G$ are joinable, $\am$ and $\bm$ are joinable and $\am\meetf\bm$ is surjective then $H'$ and $G'$ are joinable, and
  \begin{itemize}
  \item $\am\meetf \bm$ is an isomorphism from $H\gcap G$ to $H'\gcap G'$,
  \item $\am\joinf\bm$ is an isomorphism from $H\gcup G$ to $H'\gcup G'$.
  \end{itemize}
\end{lemma}
\begin{proof} Since $\algebrf{H}=\algebrf{G}$ this set is included in the domains of $\am$ and $\bm$ which are joinable, hence $\labelf{\am}=\labelf{\bm}$ and $\algebrf{H'} = \labelf{\am}(\algebrf{H}) = \labelf{\bm}(\algebrf{G}) = \algebrf{G'}$. Let $\gamma = \am\meetf \bm$ and suppose there is a $y\in \nodes{H'}\cap\arrows{G'} = \am(\nodes{H}) \cap \bm(\arrows{G})$, this is clearly included in the codomain $\carrier{H'}\cap \carrier{G'}$ of $\gamma$. But $\gamma$ is surjective, hence there is a $x$ in the domain $\carrier{H}\cap \carrier{G}$ of $\gamma$ such that $y=\gamma(x)=\am(x)=\bm(x)$, hence $\invf{\am}(y)=x\in\nodes{H}$ and $\invf{\bm}(y)=x\in\arrows{G}$, so that $x\in\nodes{H}\cap\arrows{G}=\ensvide$, a contradiction; this proves that $\nodes{H'}\cap\arrows{G'} = \ensvide$. We prove similarly that $\arrows{H'}\cap\nodes{G'} = \ensvide$.

For all $g\in \arrows{H}'\cap \arrows{G}'$, let $f=\invf{\gamma}(g)\in \arrows{H}\cap \arrows{G}$. Since $H$ and $G$ are joinable, then $\leftf{H}(f)=\leftf{G}(f)\in \nodes{H}\cap \nodes{G}$, hence \[\leftf{H}'(g) = \leftf{H}'\circ \arrows{\am}(f) = \nodes{\am}\circ \leftf{H}(f) =  \nodes{\bm}\circ \leftf{G}(f) = \leftf{G}'\circ \arrows{\bm}(f) = \leftf{G}'(g).\] Similarly we get $\rightf{H}'(g)=\rightf{G}'(g)$, hence $H'$ and $G'$ are joinable. Furthermore, for all $f\in \arrows{H}\cap \arrows{G}$ we have \[(\leftf{H'}\meetf\leftf{G'})\circ \arrows{\gamma}(f) = \leftf{H}'\circ\arrows{\am} (f) = \nodes{\alpha}\circ \leftf{H}(f) = \nodes{\gamma}\circ (\leftf{H}\meetf\leftf{G})(f).\]
Besides, $\labelf{\gamma}=\labelf{\am}$ is a $\Sigma$-isomorphism from $\algebrf{H\gcap G} = \algebrf{H}$ to $\algebrf{H'\gcap G'}=\algebrf{H'}$, and $\forall x\in (\nodes{H}\cap\nodes{G}) \cup (\arrows{H}\cap\arrows{G})$, $\labelf{\gamma}\circ(\labelf{H}\cap\labelf{G})(x) = \labelf{\am}\circ\labelf{H}(x) = \labelf{H'}\circ\am(x) = (\labelf{H'}\cap\labelf{G'})\circ\gamma(x)$,
hence $\gamma$ is an isomorphism of $H\gcap G$ to $H'\gcap G'$. 

Let $\delta=\am \joinf \bm$, this is obviously a bijective function from $\carrier{H}\cup \carrier{G}$ onto $\carrier{H'}\cup \carrier{G'}$, and for all $f\in \arrows{H}$ we have
\[(\leftf{H}'\joinf\leftf{G}')\circ\arrows{\delta}(f) = \leftf{H}'\circ\arrows{\am}(f) = \nodes{\am}\circ \leftf{H}(f) = \nodes{\delta}\circ(\leftf{H}\joinf\leftf{G})(f),\] and similarly for all $f\in\arrows{G}$ (using $\bm$), hence  $(\leftf{H}'\joinf\leftf{G}')\circ\arrows{\delta} = \nodes{\delta}\circ(\leftf{H}\joinf\leftf{G})$. We also have $\labelf{\delta}=\labelf{\gamma}$ is a $\Sigma$-isomorphism from $\algebrf{H\gcup G} = \algebrf{H}$ to $\algebrf{H'\gcup G'}=\algebrf{H'}$, and $\forall x\in (\nodes{H}\cup\nodes{G}) \cup (\arrows{H}\cup\arrows{G})$, if $x\in\nodes{H}\cup\arrows{H}$ then $\labelf{\delta}\circ(\labelf{H}\cup\labelf{G})(x) = \labelf{\am}\circ\labelf{H}(x) = \labelf{H'}\circ\am(x) = (\labelf{H'}\cup\labelf{G'})\circ\delta(x)$; otherwise $x\in\nodes{G}\cup\arrows{G}$ and similarly $\labelf{\delta}\circ(\labelf{H}\cup\labelf{G})(x) = (\labelf{H'}\cup\labelf{G'})\circ\delta(x)$ which proves that $\labelf{\delta}\circ(\labelf{H}\cup\labelf{G}) = (\labelf{H'}\cup\labelf{G'})\circ\delta$ and hence that $\delta$ is an isomorphism from $H\gcup G$ to $H'\gcup G'$.
\end{proof}
\begin{corollary}\label{cr-isojoins}
  If $\forall i\in\set{1,\ldots,n}$, $\am_i$ is an isomorphism from $G_i$ to $G'_i$ such that $\forall j\in \set{1,\ldots,n}$, $G_i$ and $G_j$ are joinable, $\am_i$ and $\am_j$ are joinable and $\am_i\meetf\am_j$ is surjective then $\bigjoinf_{i=1}^n\am_i$ is an isomorphism from $\Gcup_{i=1}^n G_i$ to $\Gcup_{i=1}^n G'_i$.
\end{corollary}
\begin{corollary}\label{cr-groupmeet}
  If $H$ and $G$ are joinable graphs then \[\Aut{H\gcup G}{H,G}\subseteq \Aut{H\gcup G}{H\gcap G}.\]
\end{corollary}
\begin{proof}
  For all $\sigma\in\Aut{H\gcup G}{H,G}$, let $\am=\restrf{\sigma}{H}$ and $\bm=\restrf{\sigma}{G}$. Obviously $\am$ and $\bm$ are joinable isomorphisms with $\am(H)=H$, $\bm(G)=G$ and $\am\meetf\bm = \restrfc{\sigma}{\carrier{H}\cap\carrier{G}}{\carrier{H}\cap\carrier{G}}$ is surjective since $\sigma(\carrier{H}\cap\carrier{G}) = \sigma(\carrier{H})\cap\sigma(\carrier{G}) = \carrier{H}\cap\carrier{G}$. Thus $\am\meetf\bm$ is an isomorphism from $H\gcap G$ to itself, which yields $\sigma(H\gcap G) = H\gcap G$ and therefore $\sigma\in \Aut{H\gcup G}{H\gcap G}$.  
\end{proof}

We will also need to know how the removal of objects in a graph behaves through unions and isomorphisms.
\begin{lemma}\label{lm-join-remove}
  Given two graphs $H$ and $G$, $V\subseteq \nodes{H}$, $A\subseteq \arrows{H}$ and $l$ a labelling function of $H$,
  \begin{enumerate}
  \item[\emph{(1)}] if $H$ and $G$ are joinable then $\delgraph{(H\gcup G)}{V}{A}{l}\subalg \delgraph{H}{V}{A}{l}\gcup G$, and the equality $\delgraph{(H\gcup G)}{V}{A}{l}= \delgraph{H}{V}{A}{l}\gcup G$ holds iff $V\cap\nodes{G}= A\cap \arrows{G}=\ensvide$ and $l(x)\cap\labelf{G}(x)=\ensvide$ for all $x\in\nodes{G}\cup \arrows{G}$.
  \item[\emph{(2)}] if $\am$ is an isomorphism from $H$ to $G$ then \[\am(\delgraph{H}{V}{A}{l}) = \delgraph{\am(H)}{\nodes{\am}(V)}{\arrows{\am}(A)}{\labelf{\am}\circ l\circ \invf{(\nodes{\am}\joinf \arrows{\am})}}\]
  \end{enumerate}
\end{lemma}
\begin{proof}
\noindent  (1) As $\delgraph{H}{V}{A}{l}\subalg H$ thus $\delgraph{H}{V}{A}{l}$ and $G$ are also joinable. Let $F=\delgraph{(H\gcup G)}{V}{A}{l}$ and $F'=\delgraph{H}{V}{A}{l}\gcup G$. Obviously \[\nodes{F} = (\nodes{H}\cup \nodes{G})\setminus V = (\nodes{H}\setminus V)\cup (\nodes{G}\setminus V)  \subseteq (\nodes{H}\setminus V)\cup \nodes{G} = \nodes{F}'.\] Let $f\in\arrows{F}$, if $f\in\arrows{G}$ then $f\in\arrows{F}'$, otherwise $f\in\arrows{H}\setminus A$, and since $\leftf{H}(f)\in \nodes{F}\subseteq\nodes{F}'$ and $\rightf{H}(f)\in\nodes{F}'$ then again $f\in \arrows{F}'$; hence $\arrows{F}\subseteq \arrows{F}'$.
As above, for all $x\in\nodes{F}\cup\arrows{F}$ we have 
\begin{eqnarray*}
  \labelf{F}(x) &=&  (\labelf{H}(x)\cup \labelf{G}(x))\setminus l(x)\\ &=& (\labelf{H}(x)\setminus l(x))\cup (\labelf{G}(x)\setminus l(x))\\ & \subseteq & (\labelf{H}(x)\setminus l(x))\cup \labelf{G}(x) = \labelf{F}'(x). 
\end{eqnarray*}
This proves that $F\subalg F'$. Thus $F\gcup G\subalg F'\gcup G = F'$, and since $\delgraph{H}{V}{A}{l} \subalg F$ then $F' = \delgraph{H}{V}{A}{l} \gcup G \subalg F\gcup G$, hence $F' = F\gcup G$. Therefore $F=F'$ iff $F=F\gcup G$ iff $G\subalg F$ iff $V\cap\nodes{G}= A\cap \arrows{G}=\ensvide$ and $l(x)\cap\labelf{G}(x)=\ensvide$ for all $x\in\nodes{G}\cup \arrows{G}$, this last step by Lemma \ref{lm-subremove} since $G\subalg H\gcup G$.

\noindent (2) It suffices to prove that $\am(H\sgminus l) = \am(H)\sgminus \labelf{\am}\circ l\circ \invf{(\nodes{\am}\joinf \arrows{\am})}$, $\am(F\sgminus A) = \am(F)\sgminus \arrows{\am}(A)$ and $\am(F\sgminus V) = \am(F)\sgminus \nodes{\am}(V)$ for any $F\subalg H$. We have
\begin{eqnarray*}
 \am(H\sgminus l)
  &= &\tuple{\nodes{\am}(\nodes{H}),\, \arrows{\am}(\arrows{H}),\, \algebrf{G},\, \restrf{\leftf{G}}{\arrows{\am}(\arrows{H})},\, \restrf{\rightf{G}}{\arrows{\am}(\arrows{H})},\, \labelf{\am}\circ (\labelf{H}\setminus l)\circ \invf{(\nodes{\am}\joinf \arrows{\am})} }  \\
  &= &\tuple{\nodes{\am}(\nodes{H}),\, \arrows{\am}(\arrows{H}),\, \algebrf{G},\, \leftf{G},\,\rightf{G},\, \labelf{\am}\circ \labelf{H}\circ \invf{(\nodes{\am}\joinf \arrows{\am})} \setminus  \labelf{\am}\circ l\circ \invf{(\nodes{\am}\joinf \arrows{\am})}}  \\
  &= &\tuple{\nodes{\am}(\nodes{H}),\, \arrows{\am}(\arrows{H}),\, \algebrf{G},\, \nodes{\am}\circ\leftf{H} \circ\invf{\arrows{\am}},\, \nodes{\am}\circ\rightf{H} \circ\invf{\arrows{\am}},\,\\ && \hspace*{11em} \labelf{\am}\circ \labelf{H}\circ \invf{(\nodes{\am}\joinf \arrows{\am})} \setminus  \labelf{\am}\circ l\circ \invf{(\nodes{\am}\joinf \arrows{\am})}}  \\
  &=& \am(H)\sgminus \labelf{\am}\circ l\circ \invf{(\nodes{\am}\joinf \arrows{\am})}.
\end{eqnarray*}
We may assume that $A\subseteq \arrows{F}$, then
\begin{eqnarray*}
  \am(F\sgminus A)
  &=& \am\big(\tuple{\nodes{F},\, \arrows{F}\setminus A,\, \algebrf{H},\, \restrf{\leftf{F}}{A},\, \restrf{\leftf{F}}{A},\, \restrf{\labelf{F}}{\nodes{F}\cup \arrows{F} \setminus A}}\big)\\
  &=& \tuple{\nodes{\am}(\nodes{F}),\, \arrows{\am}(\arrows{F})\setminus \arrows{\am}(A),\, \algebrf{G},\, \restrfc{\leftf{G}}{\arrows{\am}(A)}{\nodes{\am}(\nodes{F})},\, \restrfc{\leftf{G}}{\arrows{\am}(A)}{\nodes{\am}(\nodes{F})},\\ && \hspace*{11em} \labelf{\am}\circ \labelf{F}\circ \restrf{\invf{(\nodes{\am}\joinf \arrows{\am})}}{\nodes{\am}(\nodes{F})\cup \arrows{\am}(\arrows{F}) \setminus \arrows{\am}(A)}}\\
  &=& \am(F)\sgminus \arrows{\am}(A).
\end{eqnarray*}
 Finally, we assume that $V\subseteq \nodes{F}$, then
\begin{eqnarray*}
  \am(F\sgminus V) & = & \am(\sgen{F}{\nodes{F}\setminus V})\\
 &=& \am(\tuple{\nodes{F}\setminus V,\, B,\, \algebrf{F},\ \restrfc{\leftf{F}}{B}{\nodes{F}\setminus V},\ \restrfc{\rightf{F}}{B}{\nodes{F}\setminus V},\, \restrf{\labelf{F}}{(\nodes{F}\setminus V)\cup B}})\\
  &=&  \tuple{W,\, \arrows{\am}(B),\, \algebrf{G},\,  \restrfc{\leftf{G}}{\arrows{\am}(B)}{W},\, \restrfc{\rightf{G}}{\arrows{\am}(B)}{W},\, \labelf{\am}\circ \labelf{F}\circ \restrf{\invf{(\nodes{\am}\joinf \arrows{\am})}}{W\cup \arrows{\am}(B)}}
\end{eqnarray*}
where $B=\invf{\leftf{F}}(\nodes{F}\setminus V) \cap \invf{\rightf{F}}(\nodes{F}\setminus V)$ and $W = \nodes{\am}(\nodes{F})\setminus \nodes{\am}(V)$. But \[\am(F)\sgminus\nodes{\am}(V) = \sgen{\am(F)}{W} = \tuple{W,\, C,\, \algebrf{G},\ \restrfc{\leftf{G}}{C}{W},\ \restrfc{\rightf{G}}{C}{W},\, \restrf{\labelf{\am}\circ \labelf{F}\circ \invf{(\nodes{\am}\joinf \arrows{\am})}}{W\cup C}}\] where $C= \invf{\leftf{G}}(W) \cap \invf{\rightf{G}}(W)$. We have $\leftf{H}\circ \invf{\arrows{\am}} = \invf{\nodes{\am}}\circ \leftf{G}$, hence their inverse functions on sets are equal, i.e., \[\arrows{\am}\circ \invf{\leftf{H}} = \invf{(\leftf{H}\circ \invf{\arrows{\am}})} = \invf{(\invf{\nodes{\am}}\circ \leftf{G})} = \invf{\leftf{G}}\circ \nodes{\am}.\] Similarly we get $\arrows{\am}\circ \invf{\rightf{H}} = \invf{\rightf{G}}\circ \nodes{\am}$, hence 
\begin{eqnarray*}
  \arrows{\am}(B) 
  &=& \arrows{\am}\big(\invf{\leftf{H}}(\nodes{F}\setminus V) \cap \invf{\rightf{H}}(\nodes{F}\setminus V)\big)\ \ \text{(since $F\subalg H$)}\\
  &=& \invf{\leftf{G}}\big(\nodes{\am}(\nodes{F}\setminus V)\big) \cap \invf{\rightf{G}}\big(\nodes{\am}(\nodes{F}\setminus V)\big) \\
  &=& \invf{\leftf{G}}(W) \cap \invf{\rightf{G}}(W) \\ &=& C,
\end{eqnarray*}
 which proves that $\am(F\sgminus V) = \am(F)\sgminus\nodes{\am}(V)$.
\end{proof}

\section{Rules}\label{sec-rules}

In general, a rewrite rule has a left-hand side $L$ to be matched in
an input graph $G$, and a right-hand part $R$ that should match in
the rewritten graph (the output). In some algebraic approaches such as
the Double-Pushout~\cite{EhrigPS73}, the graphs $L$ and $R$ have a
common part $K$ which matches in both the input and the output
graphs. All items in the input graph that are
matched by $L$ but not by $K$ are removed from the input, just
as every item in the output graph that is matched by $R$ and not by
$K$ is added to the input. In the context of parallel graph
rewriting, different rules may overlap or, more precisely, different
matchings (possibly of the same rule) may overlap and thus disagree on
what should be removed or preserved from input to output.

It is therefore convenient to allow for some flexibility in the rules,
in order to minimize the possible conflicts. In the definition below,
we will provide the possibility to express, within a rule, the fact
that every item of a graph that is not removed by a rule may not be
preserved in the output, and thus can be removed by another
rule. Hence we make a clear distinction between the graph $K$, a
subgraph of $L$ ($ K\subalg L$) that specifies what is not removed
and the common part of $L$ and $R$, i.e., $L\gcap R$ (called $M$ in
Section~\ref{sec-intro}), which specifies what ought to be preserved
in the output since it belongs to $R$.

\begin{center}
  \begin{tikzpicture}
  \draw (-2,0) circle (3cm and .5cm);
  \draw (2,0) circle (3cm and .5cm);
  \draw (0,-0.37) .. controls (-5,-0.85) and (-5, 0.85) .. (0,0.37);
  \draw (0,0) node {$L\gcap R$};
  \draw (-3,0) node {$K$};
  \draw (-4.5,0) node {$L$};
  \draw (4,0) node {$R$};
\end{tikzpicture}
\end{center}

\begin{definition}[rules, matchings]\label{def-rules}
  We assume a set $\Vars$ disjoint from $\Sigma$, whose elements are called \emph{variables}. %
 For any finite $X\subseteq \Vars$, a \emph{$(\Sigma,X)$-graph} is a
 graph $G$ such that $\algebrf{G} = \Termsig{X}$ (the $\Sigma$-term
 algebra, see e.g. \cite[p. 49]{BN98}), and the sets $\nodes{G}$,
 $\arrows{G}$ and $\labelf{G}(x)$ for all
 $x\in\nodes{G}\cup\arrows{G}$ are finite. A \emph{$\Sigma$-graph} is
 a $(\Sigma,\ensvide)$-graph. Let \[\Var{G}\ =\
   \bigcup_{x\in\nodes{G}\cup\arrows{G}}\Bigg(\bigcup_{t\in\labelf{G}(x)}
   \Var{t}\Bigg),\]
where $\Var{t}$ is the set of variables occurring in $t$, see
\cite[p. 37]{BN98}.

A \emph{rule} $r$ is a triple $\tuple{L,K,R}$ of $(\Sigma,X)$-graphs such that $L$ and $R$ are joinable, $L\gcap R\subalg K\subalg L$ and $\Var{L} = X$. Note that this implies that $\Var{R}\subseteq\Var{L}$, $R$ and $K$ are joinable and $R\gcap K = L\gcap R$.

A \emph{matching} of $r$ in a $\Sigma$-graph $G$ is a matching $\mm$ of $L$ in $G$ such that 
$\forall x\in \nodes{K}\cup\arrows{K}$, $\labelf{\mm}(\labelf{L}(x)\setminus \labelf{K}(x))\cap \labelf{\mm}(\labelf{K}(x)) = \ensvide$ (or equivalently $\labelf{\mm}(\labelf{L}(x)\setminus \labelf{K}(x)) = \labelf{\mm}(\labelf{L}(x))\setminus \labelf{\mm}(\labelf{K}(x))$). Note that $\labelf{\mm}$ may not be injective; this last condition is therefore necessary to separate the labels in $G$ that should be removed from those that should be preserved by a rewriting step. We denote $\Matches{r}{G}$ the set of all matchings of $r$ in $G$ (they all have domain $\carrier{L}$). 

We consider finite sets $\R$ of rules such that $\forall r,r'\in\R$, if $\tuple{L,K,R} = r \neq r' = \tuple{L',K',R'}$ then $\nodes{L}\cup\arrows{L} \neq \nodes{L}'\cup\arrows{L}'$, so that $\carrier{L}\neq\carrier{L'}$ hence $\Matches{r}{G}\cap \Matches{r'}{G} = \ensvide$ for any $\Sigma$-graph $G$; we then write $\Matches{\R}{G}$ for $\biguplus_{r\in\R}\Matches{r}{G}$. For any $\mm\in \Matches{\R}{G}$ there is a unique rule $\rulem{\mm}\in\R$ such that $\mm\in \Matches{\rulem{\mm}}{G}$, and its components are denoted $\rulem{\mm} = \tuple{\Lg{\mm}, \Kg{\mm}, \Rg{\mm}}$.
\end{definition}

\begin{example}\label{ex-rule}
  A rule $\tuple{L,K,R}$ may be specified as
  $\textcolor{gray}{L}\Longrightarrow R$, where the graph $L$ is depicted in
  gray and its subgraph $K$ in black. Consider for instance
\[\raisebox{-1.2ex}{\begin{tikzpicture}
\graph[grow right=2cm] {a[as={$z,\set{u}$}, gray] <-[gray, edge
  label={$g,\ensvide$}]
  b[as={$x,\set{u}$}] ->[edge label={$f,\ensvide$}] c[as={$y,\set{\textcolor{gray}{u},v}$}]};
\end{tikzpicture}}
\Longrightarrow
\raisebox{-1.2ex}{\begin{tikzpicture}
\graph[grow right=2cm] {a[as={$z',\set{v}$}] <-[edge label={$g',\ensvide$}]
  b[as={$x,\set{a}$}] ->[edge label={$f,\ensvide$}] c[as={$y,\set{s(v)}$}]}; 
\end{tikzpicture}}\]
where a comma separates every vertex and arrow $x$ from its label
$\labelf{L}(x)$ or $\labelf{R}(x)$, $u$ and $v$ are variables, $a\in
\Sigma$ is a constant and $s\in \Sigma$ has arity 1.  Here, the explicit, non graphical representation of $L$ is given by $\nodes{L}=\set{x,y,z}$, $\arrows{L}=\set{f,g}$, $\algebrf{L} = \Termsig{\set{u,v}}$, $\leftf{L} = \tuple{\set{\tuple{f,x},\tuple{g,x}},\nodes{L}}$, $\rightf{L} = \tuple{\set{\tuple{f,y},\tuple{g,z}},\nodes{L}}$ and $\labelf{L} = \tuple{\set{\tuple{x,\set{u}}, \tuple{y,\set{u,v}},\tuple{z,\set{u}},\tuple{f,\ensvide},\tuple{g,\ensvide}},\Part{\algebrf{L}}}$, hence we opt for graphical representations of graphs. Then we see that \[K= \raisebox{-1.2ex}{\begin{tikzpicture}
\graph[grow right=2cm] {b[as={$x,\set{u}$}] ->[edge label={$f,\ensvide$}] c[as={$y,\set{v}$}]};
\end{tikzpicture}}\text{ and } L\gcap R = \raisebox{-1.2ex}{\begin{tikzpicture}
\graph[grow right=2cm] {b[as={$x,\ensvide$}] ->[edge label={$f,\ensvide$}] c[as={$y,\ensvide$}]};
\end{tikzpicture}}.\]

We now consider the following $\Sigma$-graph:
\[G = \raisebox{-1.2ex}{\begin{tikzpicture}
\graph[grow right=3cm] {a[as={$3,\set{b,s(b)}$}] <-[edge label={$5,\set{b}$}]
  b[as={$1,\set{b}$}] ->[edge label={$4,\set{a}$}] c[as={$2,\set{a,b}$}]};
\end{tikzpicture}}\]
where $b\in\Sigma$ is another constant. Then there is a matching $\mm$
from the rule above to $G$, given by the relation $\set{\tuple{x,1},\,
  \tuple{y,2},\, \tuple{z,3},\, \tuple{f,4},\, \tuple{g,5},\,
  \tuple{u,b},\, \tuple{v,a}}$. Note that $\nodes{\mm}$ and
$\arrows{\mm}$ are injective and that
\[\labelf{\mm}(\labelf{L}(y)\setminus \labelf{K}(y))\cap \labelf{\mm}(\labelf{K}(y))
= \labelf{\mm}(\set{u})\cap \labelf{\mm}(\set{v}) = \set{b}\cap\set{a}
= \ensvide.\]
\end{example}

One essential feature of parallel graph rewriting is that we should consider the simultaneous use of all elements of a set of matchings of the rules in a given graph. This of course is only possible if this set is finite, which is always true by virtue of Definition \ref{def-rules}.
  
\begin{theorem}\label{thm-fini}
  $\Matches{r}{G}$ is finite.
\end{theorem}
\begin{proof}
  Let $L$ be the left part of $r$ and $X=\Var{L}$. All elements $\mm\in\Matches{r}{G}$ can be obtained as $\mm = \nodes{\mm}\joinf \arrows{\mm}\joinf \labelf{\mm}$. $\nodes{\mm}$ belongs to the finite set of functions from $\nodes{L}$ to $\nodes{G}$ and $\arrows{\mm}$ to the finite set of functions from $\arrows{L}$ to $\arrows{G}$. Since $\Termsig{X}$ is free with generating set $X$ in the class of $\Sigma$-algebras, every $\Sigma$-homomorphism $\labelf{\mm}$ is determined by $\restrf{\labelf{\mm}}{X}$. For every $v\in X$ there is an $x\in\nodes{L}\cup\arrows{L}$ and a $t\in\labelf{L}(x)$ such that $v\in\Var{t}$, and since $\labelf{\mm}(t)\in \labelf{G}(\mm(x))$ then $\labelf{\mm}(v)$ belongs to the set of subterms of the elements of $\labelf{G}(\mm(x))$, which is finite. Hence there is a finite set of possible functions $\restrf{\labelf{\mm}}{X}$.
\end{proof}

This property trivially extends to $\Matches{\R}{G}$ for any finite set $\R$ of rules. This of course explains why we have chosen the algebra of ground terms in the rewritten graphs. In practice it is often necessary to allow other algebras, e.g., the additive algebra of integers, and it is then possible to recover finiteness by imposing ad-hoc restrictions on the rules' labels, e.g., we cannot allow the term $x+y$ (where $x$ and $y$ are variables) since it has infinite matchings with any integer. 

In the sequel we will use the standard identification of \emph{substitutions} $\restrf{\labelf{\mm}}{X}$ to their homomorphic extensions $\labelf{\mm}$, so that all matchings of a rule have finite domains.

A rewriting step may involve the creation of new vertices in a graph,
corresponding to the vertices of a rule that have no match in the
input graph, i.e., those in $\nodes{R}\setminus\nodes{L}$ (or
similarly may create new arrows). These vertices should really be new,
not only different from the vertices of the original graph but also
different from the vertices created by other rewritings (corresponding
to other matchings in the graph). This is computationally easy to do
but not that easy to formalize in an abstract way. The notion of
\emph{renaming} of a rule is not adapted to parallel rewriting since
any rule has infinitely many renamings. We choose to reuse the
vertices $x$ from $\nodes{R}\setminus\nodes{L}$ by \emph{indexing}
them with any relevant matching $\mm$, each time yielding a new vertex
$\tuple{x,\mm}$ which is obviously different from any new vertex
$\tuple{x,\nm}$ for any other matching $\nm\neq\mm$, and also from any
vertex of $G$ since $\mm$ depends\footnote{$\carrier{G}$ is the
  codomain of $\mm$, hence $\tuple{x,\mm}\not\in\carrier{G}$ by the
  axiom of regularity from set theory.} on $G$.

\begin{definition}[graph $\RImg{G}{\mm}$ and matching $\liftm{\mm}$]\label{def-RImg}  
For any rule $r=\tuple{L,K,R}$, $\Sigma$-graph $G$ and $\mm\in\Matches{r}{G}$ we define a $\Sigma$-graph $\RImg{G}{\mm}$ together with a matching $\liftm{\mm}$ of $R$ in $\RImg{G}{\mm}$. We first define the sets 
$\RImg{\nodes{G}}{\mm} \defeq \nodes{\mm}(\nodes{R}\cap \nodes{K})\uplus ((\nodes{R} \setminus \nodes{K}) \times \set{\mm})$ and $\RImg{\arrows{G}}{\mm} \defeq \arrows{\mm}(\arrows{R}\cap \arrows{K})\uplus ((\arrows{R} \setminus \arrows{K}) \times \set{\mm})$, which are finite. Next we define $\liftm{\mm}$ by: $\liftm{\labelf{\mm}} \defeq \labelf{\mm}$, $\liftm{\nodes{\mm}}$ is the function from $\nodes{R}$ to $\RImg{\nodes{G}}{\mm}$ such that $\forall x\in \nodes{R}$, if $x\in\nodes{K}$ then $\liftm{\nodes{\mm}}(x) \defeq\nodes{\mm}(x)$ else $\liftm{\nodes{\mm}}(x) \defeq \tuple{x,\mm}$, and similarly $\liftm{\arrows{\mm}}$ is the function from $\arrows{R}$ to $\RImg{\arrows{G}}{\mm}$ such that $\forall f\in \arrows{R}$, if $f\in\arrows{K}$ then $\liftm{\arrows{\mm}}(f) \defeq\arrows{\mm}(f)$ else $\liftm{\arrows{\mm}}(f) \defeq \tuple{f,\mm}$. Since $\liftm{\nodes{\mm}}$ and $\liftm{\arrows{\mm}}$ are bijective, then $\liftm{\mm}$ is a matching of $R$ in the $\Sigma$-graph \[\RImg{G}{\mm}\ \defeq\ \tuple{\RImg{\nodes{G}}{\mm},\RImg{\arrows{G}}{\mm}, \Termsig{\ensvide}, \liftm{\nodes{\mm}}\circ \leftf{R}\circ \invf{\liftm{\arrows{\mm}}}, \liftm{\nodes{\mm}}\circ \rightf{R}\circ \invf{\liftm{\arrows{\mm}}}, \liftm{\labelf{\mm}}\circ \labelf{R}\circ \invf{(\liftm{\nodes{\mm}}\joinf \liftm{\arrows{\mm}})}}.\]
\end{definition}

\begin{example}\label{ex-RImg}
  Following Example \ref{ex-rule} we get \[\liftm{\mm} = \set{\tuple{x,1},\,
  \tuple{y,2},\, \tuple{z',\tuple{z',\mm}},\, \tuple{f,4},\, \tuple{g',\tuple{g',\mm}},\,
  \tuple{u,b},\, \tuple{v,a}},\] which is a matching from $R$ to the
graph
\[\RImg{G}{\mm} = \raisebox{-1.2ex}{\begin{tikzpicture}
\graph[grow right=3cm] {a[as={$\tuple{z',\mm},\set{a}$}] <-[edge label={$\tuple{g',\mm},\ensvide$}]
  b[as={$1,\set{a}$}] ->[edge label={$4,\ensvide$}] c[as={$2,\set{s(a)}$}]}; 
\end{tikzpicture}}.\]
\end{example}

By construction $\mm$ and $\liftm{\mm}$ are joinable and $\mm\meetf\liftm{\mm}$ is a matching of $R\gcap K$ in $\mm(R\gcap K)$. We now prove that the graphs $\RImg{G}{\mm}$ can be joined to $G$.

\begin{lemma}\label{lm-join}
  For every rule $r=\tuple{L,K,R}$, $\Sigma$-graph $G$ and $\mm\in\Matches{r}{G}$, the graphs $G$ and $\RImg{G}{\mm}$ are joinable and $\mm(R\gcap K) \subalg G\gcap \RImg{G}{\mm} \subg\mm(R\gcap K)$.
\end{lemma}
\begin{proof}
  As $\carrier{G}$ is the codomain of $\mm$ then as above $\arrows{G}\cap (\arrows{R}\setminus\arrows{K})\times\set{\mm} = \ensvide$, hence $\arrows{G}\cap\RImg{\arrows{G}}{\mm}= \arrows{\mm}(\arrows{R}\cap\arrows{K})$. It is similarly obvious that $\nodes{G}\cap\RImg{\arrows{G}}{\mm} = \arrows{G}\cap\RImg{\nodes{G}}{\mm} = \ensvide$. For all $f\in \arrows{R}\cap\arrows{K}$, we have by definition of $\liftm{\mm}$ that $\liftm{\arrows{\mm}}(f)= \arrows{\mm}(f)$ and $\leftf{R}(f)\in\nodes{K}\subseteq \nodes{L}$. Then
\[\begin{array}{rclr}
    \RImg{\leftf{G}}{\mm}(\arrows{\mm}(f)) &=& \RImg{\leftf{G}}{\mm}\circ\liftm{\arrows{\mm}}(f) &\\
    &=& \liftm{\nodes{\mm}}\circ \leftf{R}(f) & (\text{by definition of }\RImg{\leftf{G}}{\mm})\\
    &=& \nodes{\mm}\circ \leftf{R}(f) & (\text{by definition of $\liftm{\nodes{\mm}}$ on $\nodes{K}$})\\
    &=& \nodes{\mm}\circ \leftf{L}(f) & (\text{since } K\subalg L)\\
    &=& \leftf{G}\circ \arrows{\mm}(f) & \text{(because $\mm$ is a morphism from $L$ to $G$).}
  \end{array}\]
  Hence $\leftf{G}$ and $\RImg{\leftf{G}}{\mm}$ are joinable and similarly $\rightf{G}$ and $\RImg{\rightf{G}}{\mm}$ are joinable, which proves that $G$ and $\RImg{G}{\mm}$ are joinable. Besides, 
  \begin{eqnarray*}
    \mm(R\gcap K) &=& \tuple{\nodes{\mm}(\nodes{R}\cap\nodes{K}),\, \arrows{\mm}(\arrows{R}\cap\arrows{K}),\, \Termsig{\ensvide},\, \restrfc{\leftf{G}}{\arrows{\mm}(\arrows{R}\cap\arrows{K})}{\nodes{\mm}(\nodes{R}\cap\nodes{K})},\, \restrfc{\rightf{G}}{\arrows{\mm}(\arrows{R}\cap\arrows{K})}{\nodes{\mm}(\nodes{R}\cap\nodes{K})},\, l}\\
                  &=& \tuple{\nodes{G}\cap\RImg{\nodes{G}}{\mm},\, \arrows{G}\cap\RImg{\arrows{G}}{\mm},\, \Termsig{\ensvide},\, \leftf{G}\meetf \RImg{\leftf{G}}{\mm},\, \rightf{G}\meetf \RImg{\rightf{G}}{\mm},\, l}
  \end{eqnarray*}
where $l = \labelf{\mm}\circ (\labelf{R}\cap\labelf{K}) \circ \invf{(\nodes{\mm}\joinf \arrows{\mm})}$. Hence $\mm(R\gcap K)\subg G\gcap\RImg{G}{\mm} \subg \mm(R\gcap K)$, and furthermore $\forall y\in \nodes{\mm}(\nodes{R}\cap\nodes{K})\cup \arrows{\mm}(\arrows{R}\cap\arrows{K})$, let $x = \invf{(\nodes{\mm}\joinf \arrows{\mm})}(y)$, then
\begin{eqnarray*}
  l(y) &=& \labelf{\mm} (\labelf{R}(x)\cap\labelf{K}(x))\\
       &\subseteq & \labelf{\mm} \circ\labelf{R}(x)\cap \labelf{\mm}\circ \labelf{K}(x)\\
       &\subseteq & \liftm{\labelf{\mm}} \circ\labelf{R}(x)\cap \labelf{\mm}\circ \labelf{L}(x)\\
       &\subseteq & \RImg{\labelf{G}}{\mm}\circ(\liftm{\nodes{\mm}}\joinf \liftm{\arrows{\mm}}) (x)\cap \labelf{G}\circ(\nodes{\mm}\joinf \arrows{\mm})(x)\\
       &\subseteq & (\RImg{\labelf{G}}{\mm}\cap \labelf{G})(y)
\end{eqnarray*}
and therefore $\mm(R\gcap K)\subalg G\gcap\RImg{G}{\mm}$.
\end{proof}
\begin{corollary}\label{cor-join}
  $\forall \mm,\nm\in\Matches{r}{G}$, the graphs $\RImg{G}{\mm}$ and $\RImg{G}{\nm}$ are joinable.
\end{corollary}
\begin{proof}
  If $\mm=\nm$ this is obvious, so we assume that $\mm\neq\nm$ so that $\RImg{\arrows{G}}{\nm} \cap (\arrows{R}\setminus\arrows{K})\times\set{\mm} = \ensvide$ and $\RImg{\arrows{G}}{\mm} \cap (\arrows{R}\setminus\arrows{K})\times\set{\nm} = \ensvide$, hence $\RImg{\arrows{G}}{\mm} \cap \RImg{\arrows{G}}{\nm} = \arrows{\mm}(\arrows{R}\cap\arrows{K}) \cap \arrows{\nm}(\arrows{R}\cap\arrows{K})\subseteq \arrows{G}$, and since $\RImg{G}{\mm}$ and $\RImg{G}{\nm}$ are both joinable with $G$ then they are joinable with each other.
\end{proof}

\section{Parallel Rewriting}\label{sec-full-rew}

For any set $M\subseteq\Matches{\R}{G}$ of matchings in $G$ we wish to define a $\Sigma$-graph $G_M$ that is the result of rewriting $G$ by applying all the rules as specified by $M$, without assuming any order. Introducing $M$ as a parameter allows us to define several parallel rewrite relations, and also encompasses the case of \emph{sequential} rewriting, defined as the special case where $M$ contains a single matching. We state some properties that we may consider appropriate for $G_M$. In conformity with our framework, we assume that $G_M$ is built from $G$, hence that $G_M$ and $G$ are joinable.
\begin{itemize}
\item[(1)] For all $\mm\in M$ there is a matching $\liftm{\mm}$ of $\Rg{\mm}$ in $G_M$ joinable with $\mm$.
\item[(2)] For all $\mm\in M$ the vertices, arrows and labels that are matched to $\Lg{\mm}$ but not to $\Kg{\mm}$ should not occur in $G_M$, i.e., $\nodes{\mm}(\nodesLg{\mm}\setminus \nodesKg{\mm}) \cap \nodes{G}_M = \arrows{\mm}(\arrowsLg{\mm}\setminus \arrowsKg{\mm}) \cap \arrows{G}_M = \ensvide$ and $\forall x\in \nodes{\mm}(\nodesKg{\mm})\cup \arrows{\mm}(\arrowsKg{\mm})$, $\labelf{\mm}\circ (\labelfLg{\mm}\setminus \labelfKg{\mm})\circ \invf{(\nodes{\mm}\joinf \arrows{\mm})}(x) \cap \labelf{G}_M(x) = \ensvide$. Note that $\nodes{\mm}$ and $\arrows{\mm}$ are injective, hence $\nodes{\mm}(\nodesLg{\mm}\setminus \nodesKg{\mm}) = \nodes{\mm}(\nodesLg{\mm})\setminus \nodes{\mm}(\nodesKg{\mm})$ and $\arrows{\mm}(\arrowsLg{\mm}\setminus \arrowsKg{\mm}) = \arrows{\mm}(\arrowsLg{\mm})\setminus \arrows{\mm}(\arrowsKg{\mm})$ and similarly, by Definition \ref{def-rules}, $\labelf{\mm}(\labelfLg{\mm}(x)\setminus \labelfKg{\mm}(x)) = \labelf{\mm}(\labelfLg{\mm}(x))\setminus \labelf{\mm}(\labelfKg{\mm}(x))$.
\item[(3)] The unmatched part of $G$ should be preserved in $G_M$. However, unmatched arrows can be deleted if they are adjacent to deleted (hence matched) vertices, hence we need only preserve $G\sgminus \bigcup_{\mm\in M}\nodes{\mm}(\nodesLg{\mm})$, in the sense that we cannot remove or add anything to this graph:
\[G\sgminus \bigcup_{\mm\in M}\nodes{\mm}(\nodesLg{\mm})\ =\sgen{G_M}{\nodes{G}\setminus \bigcup_{\mm\in M}\nodes{\mm}(\nodesLg{\mm})}.\]
\item[(4)] $G_M$ should not contain anything that is not strictly needed by properties (1) and (3), i.e., $G_M\subalg G\gcup\Gcup_{\mm\in M}\liftm{\mm}(\Rg{\mm})$. Note that properties (3) and (4) imply that $G_{\ensvide}=G$, hence we need not add this simple property to the list.
\item[(5)] Rewriting in parallel two disjoint matchings is the same as rewriting them sequentially: $\forall \mm,\nm\in\Matches{\R}{G}$, if $\nodes{\mm}(\nodesLg{\mm}) \cap \nodes{\nm}(\nodesLg{\nm}) = \ensvide$ then  $G_{\set{\mm,\nm}}$ is isomorphic to $(G_{\set{\mm}})_{\set{\iota\circ\nm}}$, where $\iota$ is the canonical matching of $\nm(\Lg{\nm})$ in $G_{\set{\mm}}$ (so that $\iota\circ\nm\in\Matches{\R}{G_{\set{\mm}}}$).
\item[(6)] The construction of $G_M$ should be invariant, i.e., if $\am$ is an isomorphism from $H$ to $G$ and $M\subseteq \Matches{\R}{H}$, then $H_M \isog G_{\am\circ M}$ (note that $\am\circ M=\setof{\am\circ\mm}{\mm\in M} \subseteq \Matches{\R}{G}$).
\end{itemize}
As reasonable candidates for $G_M$ we define the following two graphs.
\begin{definition}[graphs $\PRgraphl{G}{M}$ and $\PRgraphr{G}{M}$]\label{def-parallel-rew}
 For any $\Sigma$-graph $G$ and $M\subseteq \Matches{\R}{G}$, let
  \[\PRgraphl{G}{M} \defeq \delgraph{(G \gcup \Gcup_{\mm\in
        M}\RImg{G}{\mm})}{V}{A}{l}\ \ \ \text{ and }\ \ \ \PRgraphr{G}{M} \defeq \delgraph{G}{V}{A}{l} \gcup \Gcup_{\mm\in M}\RImg{G}{\mm},\text{ where}\]
\[V = \bigcup_{\mm\in M}\nodes{\mm}(\nodesLg{\mm}\setminus \nodesKg{\mm}),\ A = \bigcup_{\mm\in M}\arrows{\mm}(\arrowsLg{\mm}\setminus \arrowsKg{\mm})\text{ and } l= \bigcup_{\mm\in M}\labelf{\mm}\circ(\labelfLg{\mm}\setminus \labelfKg{\mm})\circ\invf{(\nodes{\mm}\joinf\arrows{\mm})}.\]
\end{definition}
Note that $l$ is only defined on $\Gcup_{\mm\in M}\mm(\Kg{\mm})$; as
mentioned above $l$ is implicitly extended to the suitable domain by
mapping other vertices and arrows to $\ensvide$.

\begin{example}\label{ex-PRgraph}
  Following Examples \ref{ex-rule} and \ref{ex-RImg}, we let
  $M=\set{\mm}$ and we have $V=\set{3}$, $A=\set{5}$ and $l = \set{\tuple{1,\ensvide},\,
  \tuple{2,\set{b}},\, \tuple{3,\set{b}},\, \tuple{4,\ensvide},\,
  \tuple{5,\ensvide}}$. Hence
\[\delgraph{G}{V}{A}{l} = \raisebox{-1.2ex}{\begin{tikzpicture}
\graph[grow right=3cm] { b[as={$1,\set{b}$}] ->[edge label={$4,\set{a}$}] c[as={$2,\set{a}$}]};
\end{tikzpicture}}.\]
Note that the label $s(b)$ is removed not by $l$ but because it labels
the vertex 3, which is removed. The same is true of label $b$ of arrow
5. By computing the union with the graph $\RImg{G}{\mm}$ of Example
\ref{ex-RImg} we get
\[\PRgraphr{G}{M} = \raisebox{-1.2ex}{\begin{tikzpicture}
\graph[grow right=3cm] {a[as={$\tuple{z',\mm},\set{a}$}] <-[edge label={$\tuple{g',\mm},\ensvide$}]
  b[as={$1,\set{a,b}$}] ->[edge label={$4,\set{a}$}] c[as={$2,\set{a,s(a)}$}]}; 
\end{tikzpicture}}.\]
The reader may check that $\PRgraphl{G}{M} = \PRgraphr{G}{M}$, see
Section \ref{sec-regularity}.
\end{example}

We immediately obtain a number of properties for these graphs:
\begin{itemize}
\item $\PRgraphl{G}{M}$ and $\PRgraphr{G}{M}$ are indeed $\Sigma$-graphs, since by Lemma \ref{lm-join} and Corollary \ref{cor-join} the $\gcup$ operation is only applied on joinable graphs, and by Theorem \ref{thm-fini} the set $M$ is finite, so that the sets $\nodes{G}\cup\bigcup_{\mm\in M}\RImg{\nodes{G}}{\mm}$, $\arrows{G}\cup\bigcup_{\mm\in M}\RImg{\arrows{G}}{\mm}$ and $\labelf{G}(x)\cup\bigcup_{\mm\in M}\RImg{\labelf{G}}{\mm}(x)$ for all $x$ in any of the two previous sets, are finite.
\item $\PRgraphl{G}{M}$ and $\PRgraphr{G}{M}$ trivially fulfill property (4).
\item $\PRgraphl{G}{M}$ trivially fulfills property (2) and $\PRgraphr{G}{M}$ property (1).
\item By Lemma \ref{lm-join-remove} (1) $\PRgraphl{G}{M}\subalg \PRgraphr{G}{M}$; they can be considered as the minimal and maximal graphs that could be defined from $M$.
\end{itemize}

We now prove that the remaining properties are also fulfilled by the two graphs.
\begin{theorem}\label{thm-prop356}
For all graphs $H$ and $G$, we have that
  \begin{itemize}
  \item[\emph{(3)}] for all $M\subseteq \Matches{\R}{G}$, let $V=\bigcup_{\mm\in M}\nodes{\mm}(\nodesLg{\mm})$, then \[G\sgminus V = \sgen{\PRgraphl{G}{M}}{\nodes{G}\setminus V} =\sgen{\PRgraphr{G}{M}}{\nodes{G}\setminus V},\]
  \item[\emph{(5)}] $\forall \mm,\nm\in\Matches{\R}{G}$, if $\nodes{\mm}(\nodesLg{\mm}) \cap \nodes{\nm}(\nodesLg{\nm}) = \ensvide$ then $\PRgraphl{G}{\set{\mm,\nm}}\isog \PRgraphl{(\PRgraphl{G}{\set{\mm}})}{\set{\iota\circ\nm}}$, where $\iota$ is the canonical matching of $\nm(\Lg{\nm})$ in $G_{\set{\mm}}$, 
  \item[\emph{(6)}] if $\am$ is an isomorphism from $H$ to $G$ and $M\subseteq \Matches{\R}{H}$, then $\PRgraphl{H}{M} \isog \PRgraphl{G}{\am\circ M}$ and $\PRgraphr{H}{M} \isog \PRgraphr{G}{\am\circ M}$.
  \end{itemize}
\end{theorem}
\begin{proof}
(3) Since no arrow is added that is adjacent only to elements of $\nodes{G}\sgminus V$, we have $\sgen{\PRgraphl{G}{M}}{\nodes{G}\setminus V} \subalg \sgen{\PRgraphr{G}{M}}{\nodes{G}\setminus V} \subalg \sgen{G}{\nodes{G}\setminus V} = G\sgminus V$, hence wee need only prove $\sgen{G}{\nodes{G}\setminus V} \subalg \sgen{\PRgraphl{G}{M}}{\nodes{G}\setminus V}$, i.e., that every arrow $f\in\arrows{G}$ such that $\leftf{G}(f)$ and $\rightf{G}(f)$ both belong to $\nodes{G}\setminus V$ is also in $\PRgraphl{\arrows{G}}{M}$. Suppose this is not the case, then $f$ must have been removed, which is only possible if $f\in \bigcup_{\mm\in M}\arrows{\mm}(\arrowsLg{\mm}\setminus \arrowsKg{\mm})$, hence there is a $\mm\in M$ such that $f\in \arrows{\mm}(\arrowsLg{\mm})$, but then $\leftf{G}(f)$ and $\rightf{G}(f)$ must belong to $\nodes{\mm}(\nodesLg{\mm})$, hence to $V$, which is impossible.

(5) Let $G'=G\gcup \RImg{G}{\mm}\gcup \RImg{G}{\nm}$, $G''=G\gcup \RImg{G}{\mm}\gcup \RImg{G}{\iota\circ\nm}$, $M = \set{\mm,\nm}$ and for any matching $\tau$ of a rule of $\R$ in a graph, $V_{\tau} = \nodes{\tau}(\nodesLg{\tau}\setminus \nodesKg{\tau})$, $A_{\tau} = \arrows{\tau}(\arrowsLg{\tau}\setminus \arrowsKg{\tau})$ and $l_{\tau}= \labelf{\tau}\circ(\labelfLg{\tau}\setminus \labelfKg{\tau}) \circ\invf{(\nodes{\tau}\joinf\arrows{\tau})}$. Since $\nodes{\mm}(\nodesLg{\mm}) \cap \nodes{\nm}(\nodesLg{\nm}) = \ensvide$ then $V_{\mm}\cap V_{\nm}= \ensvide$, $A_{\mm}\cap A_{\nm}= \ensvide$, $l_{\mm}\cap l_{\nm}$ always returns $\ensvide$ and $\RImg{(\PRgraphl{G}{\set{\mm}})}{\iota\circ\nm} = \RImg{G}{\iota\circ\nm}$. It is obvious that $V_{\iota\circ\nm} = V_{\nm}$, $A_{\iota\circ\nm} = A_{\nm}$, $l_{\iota\circ\nm}$ and $l_{\nm}$ always return the same value, hence
\begin{eqnarray*}
 \PRgraphl{(\PRgraphl{G}{\set{\mm}})}{\set{\iota\circ\nm}} 
  &=& \delgraph{\Big(\big(\delgraph{(G\gcup \RImg{G}{\mm})}{V_{\mm}}{A_{\mm}}{l_{\mm}}\big)\gcup\RImg{G}{\iota\circ\nm}
      \Big)}{V_{\nm}}{A_{\nm}}{l_{\nm}}\\ 
  &=& \delgraph{\big(\delgraph{(G\gcup \RImg{G}{\mm}\gcup\RImg{G}{\iota\circ\nm})}{V_{\mm}}{A_{\mm}}{l_{\mm}}\big)}{V_{\nm}}{A_{\nm}}{l_{\nm}}
\end{eqnarray*}
because $V_{\mm}\cap \RImg{\nodes{G}}{\iota\circ\nm} = \ensvide$, $A_{\mm}\cap \RImg{\arrows{G}}{\iota\circ\nm} = \ensvide$ and $l_{\nm}$ always returns $\ensvide$ on $\RImg{\nodes{G}}{\iota\circ\nm}\cup \RImg{\arrows{G}}{\iota\circ\nm}$. Then, it is easy to see that $\PRgraphl{(\PRgraphl{G}{\set{\mm}})}{\set{\iota\circ\nm}}  = \delgraph{G''}{V_{\mm}\cup V_{\nm}}{A_{\mm}\cup A_{\nm}}{l_{\mm}\cup l_{\nm}}$.

We now define a function $\am$ from $\carrier{G'}$ to $\carrier{G''}$ by: for any $y\in\carrier{G'}$, if $y$ is of the form $\tuple{x,\nm}$ then $\am(y) = \tuple{x,\iota\circ\nm}$, otherwise $\am(y)=y$. It is obvious that $\am$ is bijective, and we prove that it is a morphism: for all $g\in\arrows{G'}$, if $g=\tuple{f,\nm}$ for some $f\in \arrowsRg{\nm}\setminus \arrowsKg{\nm}$ then $\leftf{G''}\circ\arrows{\am}(g) = \RImg{\leftf{G}}{\iota\circ\nm}\circ \liftm{(\iota\circ\nm)}(f) = \liftm{(\nodes{\iota}\circ\nodes{\nm})}\circ \leftfRg{\nm}(f) = \nodes{\am}\circ \liftm{\nodes{\nm}}\circ \leftfRg{\nm}(f) = \nodes{\am}\circ \RImg{\leftf{G}}{\nm}\circ \liftm{\arrows{\nm}}(f) = \nodes{\am}\circ \leftf{G'}(g)$, otherwise $\leftf{G''}\circ\arrows{\am}(g) = \leftf{G'}(g) = \nodes{\am}\circ\leftf{G'}(g)$, hence $\leftf{G''}\circ\arrows{\am} = \nodes{\am}\circ\leftf{G'}$. Proving that $\rightf{G''}\circ\arrows{\am} = \nodes{\am}\circ\rightf{G'}$ and $\labelf{G}''\circ (\nodes{\am}\joinf \arrows{\am}) = \labelf{\am}\circ \labelf{G}'$ is similar (note that $\labelf{\am}$ is the identity $\Sigma$-automorphism of $\Termsig{\ensvide}$). Hence $\am$ is an isomorphism from $G'$ to $G''$ and by Lemma \ref{lm-join-remove} (2) 
\begin{eqnarray*}  
  \am(\PRgraphl{G}{M}) &=& \delgraph{\am(G')}{\nodes{\am}(V_{\mm}\cup V_{\nm})}{\arrows{\am}(A_{\mm}\cup
  A_{\nm})}{\labelf{\am}\circ (l_{\mm}\cup l_{\nm})\circ \invf{(\nodes{\am}\joinf \arrows{\am})}}\\
  &=& \delgraph{G''}{V_{\mm}\cup V_{\nm}}{A_{\mm}\cup A_{\nm}}{l_{\mm}\cup l_{\nm}}\\
  &=& \PRgraphl{(\PRgraphl{G}{\set{\mm}})}{\set{\iota\circ\nm}}.
\end{eqnarray*}

(6)  We first notice that $\forall \mm\in M$, $\am\circ\mm$ is a matching of $\Lg{\mm}$ in $G$, i.e., $\am\circ\mm\in\Matches{\rulem{\mm}}{G}$ and thus $\rulem{\am\circ\mm}=\rulem{\mm}$ according to our convention on $\R$ given in Definition~\ref{def-rules}. We now build an isomorphism $\liftm{\am}$ from $H'= H\gcup \Gcup_{\mm\in M}\RImg{H}{\mm}$ to $G'= G\gcup \Gcup_{\mm\in M}\RImg{G}{\am\circ\mm}$: let $\liftm{\labelf{\am}} = \labelf{\am}$, $\liftm{\am}(x) = \am(x)$ for all $x\in\nodes{H}\cup\arrows{H}$ and $\liftm{\am}(\tuple{x,\mm}) = \tuple{x,\am\circ\mm}$ for all $\mm\in M$ and $x\in (\nodesRg{\mm}\setminus \nodesKg{\mm}) \cup (\arrowsRg{\mm}\setminus \arrowsKg{\mm})$. As 
  \begin{eqnarray*}
    \carrier{G'} &=& \carrier{\algebrf{G}}\cup \carrier{G} \cup \bigcup_{\mm\in M}(\nodesRg{\mm}\setminus \nodesKg{\mm})\times\set{\am\circ\mm} \cup (\arrowsRg{\mm}\setminus \arrowsKg{\mm}) \times\set{\am\circ\mm}\\
& = & \liftm{\am}\big(\carrier{\algebrf{H}}\cup \carrier{H} \cup \bigcup_{\mm\in M} (\nodesRg{\mm}\setminus \nodesKg{\mm})\times\set{\mm} \cup (\arrowsRg{\mm}\setminus \arrowsKg{\mm})\times\set{\mm}\big)\\
& = & \liftm{\am}(\carrier{H'})
  \end{eqnarray*}
it is obvious that $\liftm{\am}$ is bijective. Besides, by Definition \ref{def-RImg}, $\liftm{(\am\circ\mm)}$ is a matching of $\Rg{\mm}$ in $\RImg{G}{\am\circ\mm}$ such that, for all $x\in \nodesRg{\mm}\cap\nodesKg{\mm}$, $\liftm{(\am\circ\mm)}(x) = \am\circ\mm(x) = \liftm{\am}\circ\liftm{\mm}(x)$, and for all $x\in \nodesRg{\mm}\setminus \nodesKg{\mm}$, $\liftm{(\am\circ\mm)}(x) = \tuple{x,\am\circ\mm} = \liftm{\am}(\tuple{x,\mm}) = \liftm{\am}\circ\liftm{\mm}(x)$, and similarly for all $f\in\arrowsRg{\mm}$, $\liftm{(\am\circ\mm)}(f) = \liftm{\am}\circ\liftm{\mm}(f)$, hence $\liftm{(\am\circ\mm)} = \liftm{\am}\circ\liftm{\mm}$. 

We now prove that $\liftm{\am}$ is a morphism. For all $f\in \arrows{H}$ we have $\leftf{G'}\circ\liftm{\arrows{\am}}(f) = \leftf{G}\circ \arrows{\am}(f) = \nodes{\am}\circ\leftf{H}(f) = \liftm{\nodes{\am}}\circ\leftf{H'}(f)$. Then, for all $\mm\in M$ and $f\in \arrowsRg{\mm}\setminus\arrowsKg{\mm}$ we have $\RImg{\leftf{G}}{\am\circ\mm} \circ\liftm{\arrows{\am}} = \liftm{(\nodes{\am}\circ\nodes{\mm})}\circ \leftfRg{\mm} \circ \invf{\liftm{(\arrows{\am}\circ\arrows{\mm})}} \circ\liftm{\arrows{\am}} = \liftm{\nodes{\am}}\circ\liftm{\nodes{\mm}}\circ \leftfRg{\mm} \circ \invf{\liftm{\arrows{\mm}}} =  \liftm{\nodes{\am}}\circ \RImg{\leftf{H}}{\mm}$. Hence we get $\leftf{G'}\circ\liftm{\arrows{\am}} = \liftm{\nodes{\am}}\circ\leftf{H'}$ and similarly $\rightf{G'}\circ\liftm{\arrows{\am}} = \liftm{\nodes{\am}}\circ\rightf{H'}$. We also have, for all $\mu\in M$ and $x\in (\nodesRg{\mm}\setminus \nodesKg{\mm})\cup (\arrowsRg{\mm}\setminus \arrowsKg{\mm})$, that $\labelf{G}\circ(\liftm{\nodes{\am}}\joinf \liftm{\arrows{\am}})(\tuple{x,\mm}) = \labelf{G}(\tuple{x, \am\circ\mm}) = \ensvide$ since $\tuple{x,\am\circ\mm}\not\in\carrier{G}$, hence
\begin{eqnarray*}
  \labelf{G}'\circ(\liftm{\nodes{\am}}\joinf \liftm{\arrows{\am}}) &=&
\big( \labelf{G}\cup\bigcup_{\mm\in M}\RImg{\labelf{G}}{\am\circ \mm} \big)\circ(\liftm{\nodes{\am}}\joinf \liftm{\arrows{\am}}) \\
&=& \big(\labelf{G}\circ(\nodes{\am}\joinf \arrows{\am})\big) \cup \bigcup_{\mm\in M}\RImg{\labelf{G}}{\am\circ \mm} \circ(\liftm{\nodes{\am}}\joinf \liftm{\arrows{\am}})\\
&=& (\labelf{\am}\circ\labelf{H}) \cup \bigcup_{\mm\in M} \liftm{\labelf{\am}}\circ\liftm{\labelf{\mm}} \circ \labelfRg{\mm} \circ \invf{(\liftm{\nodes{\mm}}\joinf \liftm{\arrows{\mm}})}\\
&=& (\liftm{\labelf{\am}}\circ\labelf{H}) \cup \big(\liftm{\labelf{\am}}\circ \bigcup_{\mm\in M}\RImg{H}{\mm}\big) \\
&=& \liftm{\labelf{\am}}\circ\labelf{H'}
\end{eqnarray*}
which proves that $\liftm{\am}$ is an isomorphism from $H'$ to $G'$. We now let $V$, $A$ and $l$ as in Definition \ref{def-parallel-rew} for graph $H$ and set $M$, and similarly $V'$, $A'$, $l'$ for $G$ and $\am\circ M$, then \[\liftm{\nodes{\am}}(V) = \nodes{\am}\big(\bigcup_{\mm\in M}\nodes{\mm}(\nodesLg{\mm}\setminus \nodesKg{\mm})\big) = \bigcup_{\mm\in M}\nodes{\am}\circ\nodes{\mm}(\nodesLg{\mm}\setminus \nodesKg{\mm}) = V'\]
and similarly $\liftm{\arrows{\am}}(A)=A'$ and
\begin{eqnarray*}
  l' &=& \bigcup_{\mm\in M}\labelf{\am}\circ \labelf{\mm}\circ (\labelfLg{\mm}\setminus \labelfKg{\mm})\circ \invf{\big((\nodes{\am}\circ \nodes{\mm})\joinf(\arrows{\am}\circ \arrows{\mm}\big))}\\
  &=& \labelf{\am}\circ \bigcup_{\mm\in M}\labelf{\mm}\circ (\labelfLg{\mm}\setminus \labelfKg{\mm})\circ \invf{(\nodes{\mm}\joinf \arrows{\mm})}\circ \invf{(\nodes{\am}\joinf\arrows{\am})}\\
  &=& \labelf{\am}\circ l \circ \invf{(\nodes{\am}\joinf\arrows{\am})}\\
  &=& \liftm{\labelf{\am}}\circ l \circ \invf{(\liftm{\nodes{\am}}\joinf \liftm{\arrows{\am}})}
\end{eqnarray*}
since $l$ yields $\ensvide$ out of $H$. We get by Lemma \ref{lm-join-remove} (2)
\[\liftm{\am}(\PRgraphr{H}{M}) = \liftm{\am}(\delgraph{H'}{V}{A}{l}) = \delgraph{G'}{V'}{A'}{l'} = \PRgraphr{G}{\am\circ M},\]
hence a suitable restriction of $\liftm{\am}$ is an isomorphism from $\PRgraphr{H}{M}$ to $\PRgraphr{G}{\am\circ M}$, and $\liftm{\am}(\delgraph{H}{V}{A}{l}) = \delgraph{G}{V'}{A'}{l'}$, hence a restriction $\beta$ of $\liftm{\am}$ is an isomorphism from $\delgraph{H}{V}{A}{l}$ to $\delgraph{G}{V'}{A'}{l'}$. 

Finally, for all $\mm\in M$, we have $\liftm{\nodes{\am}}(\RImg{\nodes{H}}{\mm}) = \liftm{\nodes{\am}}\big(\nodes{\mm}(\nodesRg{\mm}\cap \nodesKg{\mm}) \cup (\nodesRg{\mm}\setminus \nodesKg{\mm})\times\set{\mm}\big) = \big(\nodes{\am}\circ \nodes{\mm}(\nodesRg{\mm}\cap \nodesKg{\mm})\big) \cup (\nodesRg{\mm}\setminus \nodesKg{\mm})\times \set{\am\circ \mm} = \RImg{\nodes{G}}{\am\circ\mm}$ and similarly $\liftm{\arrows{\am}}(\RImg{\arrows{H}}{\mm}) = \RImg{\arrows{G}}{\am\circ\mm}$,
hence a restriction $\gamma$ of $\liftm{\am}$ is an isomorphism from $\Gcup_{\mm\in M}\RImg{H}{\mm}$ to $\Gcup_{\mm\in M}\RImg{G}{\am\circ\mm}$. It is obvious that $\beta\meetf\gamma$ is surjective, hence by Lemma \ref{lm-isojoin} there is an isomorphism from $\PRgraphl{H}{M}$ to $\PRgraphl{G}{\am\circ M}$.
\end{proof}

The proof that $\PRgraphr{G}{M}$ also fulfills property (5) is postponed until Theorem \ref{thm-regularity} is proved. 

\begin{definition}[full parallel rewriting]
For any finite set of rules $\R$, we define two relations $\fullPRl{\R}$ and $\fullPRr{\R}$ of \emph{full parallel rewriting} between $\Sigma$-graphs by, for all $G$,
\[G\fullPRl{\R} \PRgraphl{G}{\Matches{\R}{G}}\ \text{ and }\ G\fullPRr{\R} \PRgraphr{G}{\Matches{\R}{G}}.\]  
\end{definition}

The former satisfies properties (2) to (6), the latter satisfies properties (1) and (3) to (6).

\section{Regularity}\label{sec-regularity}

It is not generally true that $\PRgraphr{G}{M}$ fulfills (2) or that
$\PRgraphl{G}{M}$ fulfills (1), since two matchings may conflict as
one removes what another retains.

\begin{example}
  We consider the rule of Example \ref{ex-rule} and the graph
\[G = \raisebox{-1.2ex}{\begin{tikzpicture}
\graph[grow right=3cm] {a[as={$3,\set{a,b}$}] <-[edge label={$5,\ensvide$}]
  b[as={$1,\set{b}$}] ->[edge label={$4,\ensvide$}] c[as={$2,\set{a,b}$}]};
\end{tikzpicture}}.\]
We have the following two matchings of the rule in $G$:
\begin{eqnarray*}
  \mm &=& \set{\tuple{x,1},\,
  \tuple{y,2},\, \tuple{z,3},\, \tuple{f,4},\, \tuple{g,5},\,
  \tuple{u,b},\, \tuple{v,a}}\\
  \nm &=& \set{\tuple{x,1},\,
  \tuple{y,3},\, \tuple{z,2},\, \tuple{f,5},\, \tuple{g,4},\,
  \tuple{u,b},\, \tuple{v,a}}
\end{eqnarray*}
hence
\begin{eqnarray*}
\RImg{G}{\mm} &=& \raisebox{-1.2ex}{\begin{tikzpicture}
\graph[grow right=3cm] {a[as={$\tuple{z',\mm},\set{a}$}] <-[edge label={$\tuple{g',\mm},\ensvide$}]
  b[as={$1,\set{a}$}] ->[edge label={$4,\ensvide$}] c[as={$2,\set{s(a)}$}]}; 
\end{tikzpicture}}\\
\RImg{G}{\nm} &=& \raisebox{-1.2ex}{\begin{tikzpicture}
\graph[grow right=3cm] {a[as={$\tuple{z',\nm},\set{a}$}] <-[edge label={$\tuple{g',\nm},\ensvide$}]
  b[as={$1,\set{a}$}] ->[edge label={$5,\ensvide$}] c[as={$3,\set{s(a)}$}]}; 
\end{tikzpicture}}
\end{eqnarray*}
meaning that we should keep vertices 2, 3 and arrows 4,5. But we also
have $V = \nodes{\mm}(\set{z})\cup\nodes{\nm}(\set{z}) = \set{2,3}$
and $A = \arrows{\mm}(\set{g})\cup\arrows{\nm}(\set{g}) = \set{4,5}$,
meaning that we should also remove these vertices and arrows. As a
result of this conflict, the graph $\PRgraphl{G}{M}$ is \[\begin{tikzpicture}
\graph[grow right=2.5cm] {a[as={$\tuple{z',\mm},\set{a}$}] <-[edge label={$\tuple{g',\mm},\ensvide$}]
  b[as={$1,\set{a,b}$}] ->[edge label={$\tuple{g',\nm},\ensvide$}]
  c[as={$\tuple{z',\nm},\set{a}$}] -!- d[as={$2,\set{a,s(a)}$}] -!- e[as={$3,\set{a,s(a)}$}]}; 
\end{tikzpicture}\]
which does not meet condition (1), and $\PRgraphr{G}{M}$ is the graph
\[ \begin{tikzpicture}
\graph[grow right=3cm] {a[as={$\tuple{z',\mm},\set{a}$}] <-[edge label={$\tuple{g',\mm},\ensvide$}]
  b[as={$1,\set{a,b}$}]->[edge label={$4,\ensvide$}]
  c[as={$2,\set{a,s(a)}$}];
 d[as={$\tuple{z',\nm},\set{a}$}] <-[edge
 label'={$\tuple{g',\nm},\ensvide$}] b -!- f[as={}] -!- e[as={$3,\set{a,s(a)}$}];
 b ->[edge label'={$5,\ensvide$}] e}; 
\end{tikzpicture}\]
which does not meet condition (2).
\end{example}

The graphs $\PRgraphl{G}{M}$ and $\PRgraphr{G}{M}$ resolve these conflicts either by prioritizing the deletion specified by the left-hand side of the rules, for $\PRgraphl{G}{M}$, or by prioritizing the preservation of the right hand side of the rules, for $\PRgraphr{G}{M}$.  Hence if we want properties (1) and (2) to be simultaneously verified on the rewritten graph we need to rule out such conflicts.

\begin{definition}[regularity]
For any $\Sigma$-graph $G$ and matchings $\mm,\nm\in\Matches{\R}{G}$, $\mm$ \emph{preserves} $\nm$ if 
$\liftm{\nm}(\Rg{\nm})\gcap \mm(\Lg{\mm})\subalg \mm(\Kg{\mm})$, i.e., the part of $G$ that is matched to $\Rg{\nm}$ cannot be removed by applying $\mm$. A set $M\subseteq\Matches{\R}{G}$ of matchings is \emph{regular} if  $\forall \mm,\nm\in M$, $\mm$ preserves $\nm$.
\end{definition}

Note that a matching $\mm$ always preserves itself since
$\liftm{\mm}(\Rg{\mm})\gcap \mm(\Lg{\mm}) =
\mm(\Rg{\mm}\gcap\Lg{\mm})\subalg \mm(\Kg{\mm})$ (see Example \ref{ex-PRgraph}). We now show that this notion is closely connected with the graphs from Definition \ref{def-parallel-rew}.

\begin{theorem}\label{thm-regularity}
  For all $\Sigma$-graphs $G$ and all $M\subseteq \Matches{\R}{G}$, \[M \text{ is regular if and only if } \ \PRgraphl{G}{M}=\PRgraphr{G}{M}.\]
\end{theorem}
\begin{proof}
  Let $H = \Gcup_{\nm\in M}\RImg{G}{\nm}$ and $V$, $A$, $l$ as in Definition \ref{def-parallel-rew}, then by Lemma \ref{lm-join-remove} (1) we have $\PRgraphl{G}{M}=\PRgraphr{G}{M}$ iff $V\cap\nodes{H}= A\cap \arrows{H}=\ensvide$ and $l(x)\cap\labelf{H}(x)=\ensvide$ for all $x\in\nodes{H}\cup \arrows{H}$. We have
\[\nodes{H}\cap V = \bigcup_{\nm\in M}\RImg{\nodes{G}}{\nm} \cap \bigcup_{\mm\in M}\nodes{\mm}(\nodesLg{\mm}\setminus \nodesKg{\mm}) = \bigcup_{\mm,\nm\in M}\RImg{\nodes{G}}{\nm} \cap \nodes{\mm}(\nodesLg{\mm}\setminus \nodesKg{\mm})\]
hence $\nodes{H}\cap V = \ensvide$ iff $\forall \mm,\nm\in M$, $\RImg{\nodes{G}}{\nm} \cap \nodes{\mm}(\nodesLg{\mm}\setminus \nodesKg{\mm}) = \ensvide$, but this is equivalent to $\RImg{\nodes{G}}{\nm} \cap \nodes{\mm}(\nodesLg{\mm})\subseteq \nodes{\mm}(\nodesKg{\mm})$. Similarly we see that $\arrows{H}\cap A = \ensvide$ iff $\forall \mm,\nm\in M$, $\RImg{\arrows{G}}{\nm} \cap \arrows{\mm}(\arrowsLg{\mm})\subseteq \arrows{\mm}(\arrowsKg{\mm})$ and that $\labelf{H}(x)\cap l(x)=\ensvide$ holds for all $x\in\nodes{H}\cup \arrows{H}$ iff $\forall \mm,\nm\in M$, $\forall x\in (\RImg{\nodes{G}}{\nm}\cap \nodes{\mm}(\nodesLg{\mm}))\cup (\RImg{\arrows{G}}{\nm}\cap \arrows{\mm}(\arrowsLg{\mm}))$, $\RImg{\labelf{G}}{\nm}(x) \cap \labelf{\mm}\circ \labelfLg{\mm}\circ \invf{(\nodes{\mm}\joinf \arrows{\mm})}(x) \subseteq \labelf{\mm}\circ \labelfKg{\mm}\circ \invf{(\nodes{\mm}\joinf \arrows{\mm})}(x)$.

Since $\RImg{G}{\nm}\gcap \mm(\Lg{\mm}) \subalg G\gcup H$ and $\mm(\Lg{\mm}) \subalg G\gcup H$ for all $\mm,\nm\in M$, then these two graphs have the same adjacencies as in $G\gcup H$, hence $\RImg{G}{\nm}\gcap \mm(\Lg{\mm}) \subalg \mm(\Lg{\mm})$ iff $\RImg{\nodes{G}}{\nm} \cap \nodes{\mm}(\nodesLg{\mm})\subseteq \nodes{\mm}(\nodesKg{\mm})$, $\RImg{\arrows{G}}{\nm} \cap \arrows{\mm}(\arrowsLg{\mm})\subseteq \arrows{\mm}(\arrowsKg{\mm})$ and $\RImg{\labelf{G}}{\nm}(x) \cap \labelf{\mm}\circ \labelfLg{\mm}\circ \invf{(\nodes{\mm}\joinf \arrows{\mm})}(x) \subseteq \labelf{\mm}\circ \labelfKg{\mm}\circ \invf{(\nodes{\mm}\joinf \arrows{\mm})}(x)$ for all $x$ as above. Hence $\PRgraphl{G}{M}=\PRgraphr{G}{M}$ iff $\forall \mm,\nm\in M$, $\RImg{G}{\nm}\gcap \mm(\Lg{\mm}) \subalg \mm(\Lg{\mm})$, that is, iff $M$ is regular since $\RImg{G}{\nm} = \liftm{\nm}(\Rg{\nm})$ by Definition \ref{def-RImg}.
\end{proof}
\begin{corollary}
   $\forall \mm,\nm\in\Matches{\R}{G}$, if $\nodes{\mm}(\nodesLg{\mm}) \cap \nodes{\nm}(\nodesLg{\nm}) = \ensvide$ then $\PRgraphr{G}{\set{\mm,\nm}}\isog \PRgraphr{(\PRgraphr{G}{\set{\mm}})}{\set{\iota\circ\nm}}$, where $\iota$ is the canonical matching of $\nm(\Lg{\nm})$ in $G_{\set{\mm}}$.
\end{corollary}
\begin{proof}
  The set $\set{\mm,\nm}$ is regular since the graphs $\liftm{\nm}(\Rg{\nm})\gcap \mm(\Lg{\mm})$ and $\liftm{\mm}(\Rg{\mm})\gcap \nm(\Lg{\nm})$ are both empty, hence by Theorem \ref{thm-prop356}
\[\PRgraphr{G}{\set{\mm,\nm}} = \PRgraphl{G}{\set{\mm,\nm}}\isog \PRgraphl{(\PRgraphl{G}{\set{\mm}})}{\set{\iota\circ\nm}} = \PRgraphr{(\PRgraphr{G}{\set{\mm}})}{\set{\iota\circ\nm}}\]
\end{proof}
\begin{corollary}
   $\PRgraphl{G}{M}$ or $\PRgraphr{G}{M}$ fulfill properties \emph{(1)} to \emph{(6)} iff $M$ is regular.
\end{corollary}
\begin{proof}
  We consider the same $H$, $V$, $A$, $l$ as in the proof above. The if part is trivial. If $\PRgraphl{G}{M}$ fulfills property (1) then $H\subalg \PRgraphl{G}{M}$ which, by Lemma \ref{lm-subremove} and then Lemma \ref{lm-join-remove} (1) entails that $\PRgraphl{G}{M} = \PRgraphr{G}{M}$, hence that $M$ is regular. If $\PRgraphr{G}{M}$ fulfills property (2) then $\PRgraphr{\nodes{G}}{M}\cap V = \PRgraphr{\arrows{G}}{M}\cap A = \ensvide$ and $\PRgraphr{\labelf{G}}{M}(x)\cap l(x) = \ensvide$ for all $x\in \nodes{H}\cup\arrows{H}$. But $H\subalg \PRgraphr{G}{M}$ hence $\nodes{H}\cap V=\arrows{H}\cap A = \ensvide$ and $\labelf{H}(x)\cap l(x) = \ensvide$ for all $x\in \nodes{H}\cup\arrows{H}$, which again by Lemma \ref{lm-join-remove} (1) entails that $\PRgraphl{G}{M} = \PRgraphr{G}{M}$ and hence that $M$ is regular.
\end{proof}

Thus, if we want properties (1) to (6) to be fulfilled by either of the full parallel rewriting relation, then we have to restrict these relations to the cases where $\Matches{\R}{G}$ is regular, which makes them identical. Equivalently, we can ensure that all properties are fulfilled by checking that $V$, $A$ and $l$ do not intersect the graph $\Gcup_{\mm\in M}\RImg{G}{\mm}$.

Of course, we may want to spare the extra work since it should be
performed for each parallel rewriting step. It would therefore be
convenient to reduce regularity to particular properties of $\R$ or
$G$ that would be easier to check. One natural idea is to only allow
sets of rules $\R$ that ensure regularity of $\Matches{\R}{G}$ for all
graphs $G$ (as in \cite[Definition~12]{EchahedM17}). It is however easy to see that this would be a drastic restriction on $\R$: given any two rules $\tuple{L_i,K_i,R_i}$ for $i=1,2$ in $\R$ and any two vertices $x_i\in\nodes{L_i}$, it is possible to build a graph $G$ and matchings $\mm_i$ from $L_i$ in $G$ such that $\mm_1(x_1)=\mm_2(x_2)$ ($G$ is the quotient of the direct sum $L_1+L_2$ by the congruence $x_1\sim x_2$). Hence if $x_1$ can be chosen in $\nodes{L_1}\setminus\nodes{K_1}$ and $x_2$ in $\nodes{K_2}\cap \nodes{R_2}$, i.e., if these sets are not empty, then $\mm_1$ does not preserve $\mm_2$ and therefore $\Matches{\R}{G}$ is not regular. Hence regularity can be guaranteed for all graphs only if all rules have $\nodes{L}\setminus\nodes{K} = \ensvide$, or if they all have $\nodes{K}\cap \nodes{R}=\ensvide$ (and this is not even a sufficient condition).

If we cannot expect that a set of rules uniformly ensures regularity, we may still curtail regularity to a restricted class of graphs $\G$. Given $\G$, are we able to characterize which sets of rules $\R$ guarantee that $\Matches{\R}{G}$ is regular  and $\PRgraphr{G}{\Matches{\R}{G}} \in\G$ for all $G\in\G$? Given a set or rules $\R$, can we determine the biggest class $\G$ with the same property? It is dubious that these questions have a useful general answer, but they may be interesting in some particular context.

\section{Parallel Rewriting modulo Automorphisms}\label{sec-Aut}

Using the full set of matchings seems exaggerated in many cases, as
illustrated below.
\begin{example}
  Consider the following rule, where all labels are empty:
\begin{center}
\includegraphics[scale=.7,clip=true, trim=5cm 10cm 5cm 12.5cm]{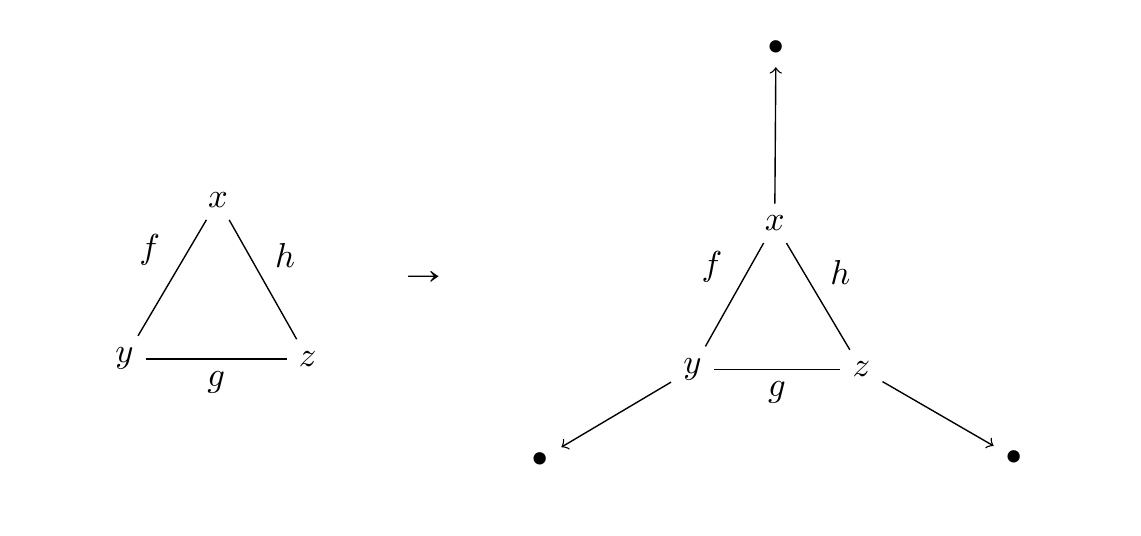}
\end{center}
here, each edge $f$, $g$, $h$ represents a pair of opposite
arrows. Obviously, this rule has six matches in a triangle, hence
\begin{center}
\includegraphics[scale=.7,clip=true, trim=5cm 17cm 5cm 4.5cm]{graphes.pdf}
\end{center}

\end{example}

We therefore wish to select a subset $M$ of $\Matches{\R}{G}$ for defining a rewriting relation that yields more natural and concise graphs. The difficulty is to maintain invariance of the result, i.e., to avoid an arbitrary choice of matchings. A key point is that we do not need to select $M$ in a deterministic way if we allow $\PRgraphr{G}{M}$ to be determined only up to isomorphism. That is, if we define a non deterministic procedure for computing $M\subseteq \Matches{\R}{G}$, and if we can ensure that for any other possible output $M'$, $\PRgraphl{G}{M'}$ is isomorphic to $\PRgraphl{G}{M}$ (or $\PRgraphr{G}{M'}$ to $\PRgraphr{G}{M}$), then the corresponding rewriting relation is deterministic up to isomorphism: it yields some undetermined element of a determined (by $\R$ and $G$) isomorphism class of graphs. 

\begin{definition}[group $\Aut{}{r}$, relation $\eqaut$]\label{def-Aut-rule}
For any rule $r=\tuple{L,K,R}$, the \emph{automorphism group} of $r$ is \[\Aut{}{r} \defeq \restrf{\Aut{L\gcup R}{L,K,R}}{L}.\]
Let $\eqaut$ be the equivalence relation on $\Matches{\R}{G}$ defined by \[\mm\eqaut\nm\text{ iff } \mm\circ\Aut{}{\rulem{\mm}} = \nm\circ\Aut{}{\rulem{\nm}}.\] The equivalence class of $\mu$ is denoted $\eqclass{\mm}{\eqaut}$. For any subset $M\subseteq  \Matches{\R}{G}$ we write $\quotient{M}{\eqaut}$ for the set $\setof{\eqclass{\mm}{\eqaut}}{\mm\in M}$ (note that $\bigcup (\quotient{M}{\eqaut})$ is a superset but may not be a subset of $M$).
\end{definition}

\begin{lemma}\label{lm-eqaut}
  $\forall \mm\in\Matches{\R}{G}$, $\eqclass{\mm}{\eqaut} = \mm\circ\Aut{}{\rulem{\mm}}$.
\end{lemma}
\begin{proof}
  If $\nm\in\eqclass{\mm}{\eqaut}$ then $\nm\circ\Aut{}{\rulem{\nm}} = \mm\circ\Aut{}{\rulem{\mm}}$ and in particular $\nm\in \mm\circ\Aut{}{\rulem{\mm}}$.

Conversely we assume that $\nm\in \mm\circ\Aut{}{\rulem{\mm}}$. From the above we see that $\Aut{}{\rulem{\mm}}$ is a subgroup of $\restrf{\Aut{\Lg{\mm}\gcup \Rg{\mm}}{\Lg{\mm}}}{\Lg{\mm}} = \Aut{}{\Lg{\mm}}$, hence $\nm$ is a matching of $\Lg{\mm}$ in $G$. But $\nm$ is also a matching of $\Lg{\nm}$ in $G$, which according to our convention on $\R$ given in Definition~\ref{def-rules} entails that $\rulem{\mm}=\rulem{\nm}$, and hence that $\Aut{}{\rulem{\mm}} = \Aut{}{\rulem{\nm}}$. Hence $\nm\circ\Aut{}{\rulem{\nm}} = \nm\circ\Aut{}{\rulem{\mm}} \subseteq \mm\circ\Aut{}{\rulem{\mm}}\circ\Aut{}{\rulem{\mm}} = \mm\circ\Aut{}{\rulem{\mm}}$. But there is a $\sigma\in \Aut{}{\rulem{\mm}}$ such that $\nm=\mm\circ\sigma$, hence such that $\mm=\nm\circ\invf{\sigma}$, which entails $\mm\in \nm\circ\Aut{}{\rulem{\mm}}$. Hence by symmetry we also have $\mm\circ\Aut{}{\rulem{\mm}} \subseteq \nm\circ\Aut{}{\rulem{\nm}}$, which proves that they are equal and hence that $\mm\eqaut \nm$ and therefore $\nm\in\eqclass{\mm}{\eqaut}$.
\end{proof}

Note that $\card{\eqclass{\mm}{\eqaut}} \leq \card{\Aut{}{\rulem{\mm}}}$ and that the equality holds if $\mm$ is injective. The more symmetric a rule is, the more matchings are likely to occur in the equivalence classes of matchings of this rule. If we could choose only one matching per equivalence class in $\Matches{\R}{G}$, we could obtain a much more concise output graph than with the rewriting relations $\fullPRl{\R}$ and $\fullPRr{\R}$ of Section \ref{sec-full-rew}. The problem of course is that selecting one matching among others may prevent the rewriting relation from being invariant. But the definition of the automorphism groups of rules has been krafted precisely so that the isomorphism class of the output graph does not depend on the choice of elements in the induced equivalence classes, which we are now in a position to prove.

\begin{theorem}
  For any $\M\in\Matches{\R}{G}$ and any minimal sets $M,N$ such that $\quotient{\M}{\eqaut} = \quotient{M}{\eqaut} = \quotient{N}{\eqaut}$, we have $\PRgraphl{G}{M}\isog\PRgraphl{G}{N}$, $\PRgraphr{G}{M}\isog\PRgraphr{G}{N}$ and $M$ is regular iff $N$ is regular.
\end{theorem}
\begin{proof}
  Since $M$ and $N$ are minimal there is a bijection $\iota$ from $M$ to $N$ such that, to every $\mm\in M$ corresponds a unique $\inN{\mm}\in N$ such that $\inN{\mm}\eqaut\mm$, hence by Lemma \ref{lm-eqaut} there is a  $\sigmA{\mm}\in\Aut{}{\rulem{\mm}}$ such that $\inN{\mm}=\mm\circ\sigmA{\mm}$, and then a $\taU{\mm}\in \Aut{\Lg{\mm}\gcup \Rg{\mm}}{\Lg{\mm},\Kg{\mm},\Rg{\mm}}$ such that $\sigmA{\mm}=\restrf{\taU{\mm}}{\Lg{\mm}}$; we let $\rhO{\mm}=\restrf{\taU{\mm}}{\Rg{\mm}}$ so that $\rhO{\mm}\in\Aut{}{\Rg{\mm}}$. Let \[\beta_{\mm} = (\liftm{\nodesinN{\mm}}\circ \invf{\nodesrhO{\mm}}\circ \invf{\liftm{\nodes{\mm}}}) \joinf (\liftm{\arrowsinN{\mm}}\circ \invf{\arrowsrhO{\mm}}\circ \invf{\liftm{\arrows{\mm}}}) \joinf \Id{\Termsig{\ensvide}},\] we now prove that this is an isomorphism from $\RImg{G}{\mm}$ to $\RImg{G}{\inN{\mm}}$. It is obvious that $\nodes{\beta}_{\mm}$ and $\arrows{\beta}_{\mm}$ are bijective functions and that $\labelf{\beta}_{\mm}$ is a $\Sigma$-isomorphism. We then see that
\[  \begin{array}{rclr}
    \nodes{\beta_{\mm}}\circ \RImg{\leftf{G}}{\mm} &=& (\liftm{\nodesinN{\mm}}\circ \invf{\nodesrhO{\mm}}\circ \invf{\liftm{\nodes{\mm}}})\circ (\liftm{\nodes{\mm}}\circ \leftfRg{\mm} \circ \invf{\arrows{\mm}}) & \text{ (by Definition \ref{def-RImg})}\\
    &=& \liftm{\nodesinN{\mm}}\circ \leftfRg{\mm} \circ \invf{\arrowsrhO{\mm}}\circ \invf{\arrows{\mm}} & \text{ (since $\invf{\rhO{\mm}}\in \Aut{}{\Rg{\mm}}$)}\\
    &=& \liftm{\nodesinN{\mm}}\circ \leftfRg{\mm} \circ \invf{\liftm{\arrowsinN{\mm}}}\circ \arrows{\beta}_{\mm} &\\
    &=& \RImg{\leftf{G}}{\inN{\mm}}\circ \arrows{\beta}_{\mm} &\text{ (since $\rulem{\inN{\mm}} = \rulem{\mm}$),}
  \end{array}\]
and similarly that $\nodes{\beta_{\mm}}\circ \RImg{\rightf{G}}{\mm} = \RImg{\rightf{G}}{\inN{\mm}}\circ \arrows{\beta}_{\mm}$. We finally see that
\begin{eqnarray*}
  \RImg{\labelf{G}}{\inN{\mm}}\circ (\nodes{\beta}_{\mm} \joinf \arrows{\beta}_{\mm}) &=& \labelfinN{\mm} \circ \labelfRg{\mm} \circ \invf{(\liftm{\nodesinN{\mm}} \joinf \liftm{\arrowsinN{\mm}})} \circ (\nodes{\beta}_{\mm} \joinf \arrows{\beta}_{\mm})\\
  &=& \labelfinN{\mm} \circ \labelfRg{\mm} \circ \invf{(\nodesrhO{\mm} \joinf \arrowsrhO{\mm})} \circ \invf{(\liftm{\nodes{\mm}} \joinf \liftm{\arrows{\mm}})}\\ 
  &=& \labelfinN{\mm} \circ \invf{\labelfrhO{\mm}}\circ \labelfRg{\mm} \circ \invf{(\liftm{\nodes{\mm}} \joinf \liftm{\arrows{\mm}})}\\
  &=& \labelfinN{\mm} \circ \invf{\labelfsigmA{\mm}}\circ \labelfRg{\mm} \circ \invf{(\liftm{\nodes{\mm}} \joinf \liftm{\arrows{\mm}})} \text{ (since $\labelfrhO{\mm} = \labelftaU{\mm} = \labelfsigmA{\mm}$)}\\
  &=& \labelf{\mm}\circ \labelfRg{\mm} \circ \invf{(\liftm{\nodes{\mm}} \joinf \liftm{\arrows{\mm}})}\\
  & = & \RImg{\labelf{G}}{\mm}\\
  & = & \labelf{\beta}_{\mm}\circ \RImg{\labelf{G}}{\mm},
\end{eqnarray*}
which shows that $\beta_{\mm}$ is an isomorphism. We next prove that this isomorphism reduces to the identity on $G\gcap \RImg{G}{\mm}$.

For all $y\in \nodes{G}\cap\RImg{\nodes{G}}{\mm}$, since by Lemma \ref{lm-join} we have $G\gcap \RImg{G}{\mm} \subg \mm(\Rg{\mm}\gcap \Kg{\mm})$, then there is a $x\in \nodesRg{\mm}\cap\nodesKg{\mm}$ such that $y = \nodes{\mm}(x) = \liftm{\nodes{\mm}}(x)$. But by Corollary \ref{cr-groupmeet} we have $\sigmA{\mm}(\Rg{\mm}\gcap \Kg{\mm}) = \Rg{\mm}\gcap \Kg{\mm}\subalg \Lg{\mm}$ hence $\invf{\nodesrhO{\mm}}(x)=\invf{\nodesigmA{\mm}}(x)$ belongs to $\nodesRg{\mm}\cap \nodesKg{\mm}$, and we get \[\nodes{\bm}_{\mm}(y) = \liftm{\nodesinN{\mm}}\circ\invf{\nodesrhO{\mm}}\circ\invf{\liftm{\nodes{\mm}}}(y) = \liftm{\nodesinN{\mm}}\circ\invf{\nodesrhO{\mm}}(x) = \nodesinN{\mm}\circ\invf{\nodesigmA{\mm}}(x) = \nodes{\mm}(x)=y.\]
We similarly get $\arrows{\bm}_{\mm}(g) = g$ for all $g\in \arrows{G}\cap\RImg{\arrows{G}}{\mm}$, which proves that $\Idm{G}\meetf\bm_{\mm} = \Idm{G\gcap \RImg{G}{\mm}} = \Idm{G}\meetf\Idm{\RImg{G}{\mm}}$ and hence that $G\gcap \RImg{G}{\mm} = G\gcap \RImg{G}{\inN{\mm}}$. 

For all $\nm\in M$, if $\nm\neq\mm$ then $\RImg{G}{\mm}\gcap \RImg{G}{\nm}\subalg G$ hence $\bm_{\mm}\meetf \bm_{\nm} = \Idm{\RImg{G}{\mm}\gcap \RImg{G}{\nm}}$ and hence $\RImg{G}{\mm}\gcap \RImg{G}{\nm} = \RImg{G}{\inN{\mm}}\gcap \RImg{G}{\inN{\nm}}$, so if we let $\bm=\bigjoinf_{\mm\in M}\bm_{\mm}$, then this is an isomorphism from $\Gcup_{\mm\in M}\RImg{G}{\mm}$ to $\Gcup_{\mm\in M}\RImg{G}{\inN{\mm}}$ by Corollary \ref{cr-isojoins}. But then
\[\Idm{G}\meetf\bm = \bigjoinf_{\mm\in M}(\Idm{G}\meetf\bm_{\mm}) = \bigjoinf_{\mm\in M}(\Idm{G}\meetf\Idm{\RImg{G}{\mm}}) = \Idm{G}\joinf \Idm{\Gcup_{\mm\in M}\RImg{G}{\mm}} = \Idm{G\gcap\Gcup_{\mm\in M}\RImg{G}{\mm}}\]
and, by Lemma \ref{lm-isojoin}, $\delta = \Idm{G}\joinf\bm$ is an isomorphism from $G\gcup\Gcup_{\mm\in M}\RImg{G}{\mm}$ to $G\gcup\Gcup_{\mm\in M}\RImg{G}{\inN{\mm}}$. We also have by Definition \ref{def-Aut-rule} that, for all $\mm\in M$, $\sigmA{\mm}(\Lg{\mm})=\Lg{\mm}$ and $\sigmA{\mm}(\Kg{\mm})=\Kg{\mm}$, hence
\[\nodesinN{\mm}(\nodesLg{\mm}\setminus \nodesKg{\mm}) = \nodes{\mm}\circ\nodes{\sigmA{\mm}}(\nodesLg{\mm})\setminus \nodes{\mm}\circ\nodesigmA{\mm}(\nodesKg{\mm}) = \nodes{\mm}(\nodesLg{\mm})\setminus \nodes{\mm}(\nodesKg{\mm}) = \nodes{\mm}(\nodesLg{\mm}\setminus \nodesKg{\mm}),\]
similarly $\arrowsinN{\mm}(\arrowsLg{\mm}\setminus \arrowsKg{\mm}) =  \arrows{\mm}(\arrowsLg{\mm}\setminus \arrowsKg{\mm})$ and 
\begin{eqnarray*}
\labelfinN{\mm}\circ (\labelfLg{\mm}\setminus \labelfKg{\mm}) \circ \invf{(\nodesinN{\mm} \joinf \arrowsinN{\mm})} &= & 
\labelf{\mm}\circ \big((\labelfsigmA{\mm}\circ\labelfLg{\mm})\setminus (\labelfsigmA{\mm}\circ\labelfKg{\mm})\big) \circ \invf{(\nodesinN{\mm} \joinf \arrowsinN{\mm})} \\
 & = & \labelf{\mm}\circ (\labelfLg{\mm}\setminus \labelfKg{\mm})\circ (\nodesigmA{\mm}\joinf \arrowsigmA{\mm}) \circ \invf{(\nodesinN{\mm} \joinf \arrowsinN{\mm})} \\
 & = & \labelf{\mm}\circ (\labelfLg{\mm}\setminus \labelfKg{\mm})\circ \invf{(\nodes{\mm} \joinf \arrows{\mm})} 
\end{eqnarray*}
hence if we let $V$, $A$ and $l$ as in Definition \ref{def-parallel-rew} for $M$, then they also work for $N$, i.e., we have \[V = \bigcup_{\mm\in M}\nodes{\mm}(\nodesLg{\mm}\setminus \nodesKg{\mm}) = \bigcup_{\mm\in M}\nodesinN{\mm}(\nodesLg{\mm}\setminus \nodesKg{\mm}) = \bigcup_{\nm\in N}\nodes{\nm}(\nodesLg{\nm}\setminus \nodesKg{\nm}),\] and similarly for $A$ and $l$. Besides, they are sets of vertices or arrows, or labelling function of $G$, hence by Lemma \ref{lm-join-remove} (2)
\begin{eqnarray*}
\delta(\PRgraphl{G}{M}) &=& \delgraph{\delta(G\gcup\Gcup_{\mm\in M}\RImg{G}{\mm})}{\nodes{\delta}(V)}{\arrows{\delta}(A)}{\labelf{\delta}\circ l\circ \invf{(\nodes{\delta}\joinf \arrows{\delta})}}\\
&=& \delgraph{\big(G\gcup\Gcup_{\mm\in M}\RImg{G}{\inN{\mm}} \big)}{\Id{\nodes{G}}(V)}{\Id{\arrows{G}}(A)}{\Id{\Termsig{\ensvide}}\circ l\circ \invf{(\Id{\nodes{G}}\joinf \Id{\arrows{G}})}}\\
&=& \delgraph{\big(G\gcup\Gcup_{\mm\in M}\RImg{G}{\inN{\mm}}\big)}{V}{A}{l}\\
&=& \PRgraphl{G}{N},
\end{eqnarray*}
hence $\delta$ is an isomorphism from $\PRgraphl{G}{M}$ to $\PRgraphl{G}{N}$.
Next, we let $G'=\delgraph{G}{V}{A}{l}$ so that $\PRgraphr{G}{M} = G'\gcup \Gcup_{\mm\in M}\RImg{G}{\mm}$ and $\PRgraphr{G}{N} = G'\gcup \Gcup_{\mm\in M}\RImg{G}{\inN{\mm}}$. Since $G'\subalg G$ then \[\Idm{G'}\meetf\bm = \Idm{G'}\meetf(\Idm{G}\meetf\bm) = \Idm{G'}\meetf \Idm{G\gcap\Gcup_{\mm\in M}\RImg{G}{\mm}} = \Idm{G'\gcap\Gcup_{\mm\in M}\RImg{G}{\mm}}\] hence, by Lemma \ref{lm-isojoin}, $\Idm{G'}\joinf\bm$ is an isomorphism from $\PRgraphr{G}{M}$ to $\PRgraphr{G}{N}$.

If $M$ is regular then $\PRgraphl{G}{M} = \PRgraphr{G}{M}$ by Theorem \ref{thm-regularity}, hence $\PRgraphl{G}{N}\isog \PRgraphr{G}{N}$. Since $\PRgraphl{G}{N}\subalg \PRgraphr{G}{N}$ and these graphs are finite, it is obvious that $\PRgraphl{G}{N}= \PRgraphr{G}{N}$, hence $N$ is regular by Theorem \ref{thm-regularity}.
\end{proof}

\begin{definition}[parallel rewriting modulo automorphisms]
For any finite set of rules $\R$, we define two relations $\PRautl{\R}$ and $\PRautr{\R}$ of \emph{parallel rewriting modulo automorphisms} between $\Sigma$-graphs by, for all $G$,
\[G\PRautl{\R} \PRgraphl{G}{M}\ \text{ and }\ G\PRautr{\R} \PRgraphr{G}{M}\]
where $M$ is any minimal set such that $\quotient{M}{\eqaut} = \quotient{\Matches{\R}{G}}{\eqaut}$.
\end{definition}
 
It is obvious from Section \ref{sec-regularity} that the two relations
are identical up to isomorphisms iff $M$ is regular. A definition of
port-graph parallel rewriting up to automorphisms has been introduced
in \cite[Definition~20]{EchahedM17}, but it was limited to
so-called symmetric rules, while our definition is more general since
it applies to all rules without restrictions.  

\section{An Example: Conway's Game of Life} \label{sec-example}

Conway's game of life \cite{Gardner70} is played on a
square grid where cells live or die according to their number of
living neighbours in the grid, following 3 deterministic rules (given
below).  These are intrinsically parallel in the sense that ``all
births and deaths occur \emph{simultaneously}'' \cite{Gardner70}, even
though the neighborhoods overlap: one cannot apply the rules sequentially and
expect the result to be independent of the order in which they are
applied, without using some trick (one was proposed by Conway, see
\cite{Gardner70}). It is customary to use the states of the current
generation to compute the states of the next generation, which
requires to keep two bits of information in each cell, rather than one
(alive or dead). It would be quite unnatural to formalize such tricks
with a sequential rewrite system.

In contrast, the fact that we can compute all possible matchings, even
when they overlap, before they are all used to simulatneously apply the
transformations specified by their respective rules, makes it very
natural to represent the computation of the next generation as a
single step of parallel rewriting. This is what we now illustrate.

\begin{itemize}
\item We start with the rule of death by overpopulation, which states that a
  living cell dies if it has at least 4 living neighbours. In other
  words, a cell dies if its neighborhood contains a subgraph of 4
  living cells, which is equivalent to saying that a graph of four
  living cells matches in the neighborhood (since a matching is
  injective on vertices). Hence this exactly corresponds to the
  following rule:

\[ \raisebox{-8ex}{
    \begin{tikzpicture}
      \graph[grow right=2cm] {
        1[as={$\set{1}$}] -!- 2[as={}] -!- 3[as={$\set{1}$}];
        4[as={}] -!- 1 -- 5[as={$x,\set{\textcolor{gray}{1}}$}] -- 3;
        7[as={$\set{1}$}] -- 5 -!- 8[as={}] -!-
        9[as={$\set{1}$}] -- 5;
      };
    \end{tikzpicture}}
  \ \ \Longrightarrow \ \ x,\set{0}
\]
\item We next consider the rule of birth: a dead cell becomes alive if
  it has exactly 3 living neighbours. The matching of a graph with
  just 3 living cells in the neighborhood is not precise enough
  here: we need to express the accuracy required for birth. However,
  knowing that a cell has 8 neighbours, we can reach this accuracy by
  stating that, among these 8, at least 3 should be alive and at least
  5 should be dead, which yields the rule:
\[ \raisebox{-8ex}{
    \begin{tikzpicture}
      \graph[grow right=2cm] {
        1[as={$\set{1}$}] -!- 2[as={$\set{1}$}] -!- 3[as={$\set{1}$}];
        4[as={$\set{0}$}] --
        5[as={$x,\set{\textcolor{gray}{0}}$}] -- 6[as={$\set{0}$}];
        1 -- 5 -- 3;
        7[as={$\set{0}$}] -- 5 -- 8[as={$\set{0}$}] -!-
        9[as={$\set{0}$}];
        2 -- 5 -- 9;
      };
    \end{tikzpicture}}
  \ \ \Longrightarrow \ \ x,\set{1}
\]
Of course, this works only for cells with exactly 8 neighbours, hence
not for borders or corners of a rectangular grid, see the discussion below.

\item We finally consider the rule of death by isolation: a living
  cell dies if it has at most 1 living neighbour. As above we have to
  assume that a cell has 8 neighbours in order to equivalently state
  that a cell dies if it has at least 7 dead neighbours, as expressed
  by the rule:
\[ \raisebox{-8ex}{
    \begin{tikzpicture}
      \graph[grow right=2cm] {
        1[as={$\set{0}$}] -!- 2[as={$\set{0}$}] -!- 3[as={$\set{0}$}];
        4[as={$\set{0}$}] -!- 1 -- 5[as={$x,\set{\textcolor{gray}{1}}$}] -- 2;
        4 -- 5 -- 3;
        7[as={$\set{0}$}] -- 5 -- 8[as={$\set{0}$}] -!-
        9[as={$\set{0}$}];
        5--9;
      };
    \end{tikzpicture}}
  \ \ \Longrightarrow \ \ x,\set{0}
\]
\end{itemize}

Hence we are able to express the game of life with the set $\R$
consisting of the 3 simple rules above. Of course, if we want our program
to work in the standard way on rectangular grids, we need to do more
than this. We would first need to assume a way of determining the
exact number of neighbours; one way to do this is to include a
constant in the labels, say $r$ for regular (8 neighbours), $b$ for
border (5 n.) and $c$ for corner (3 n.). We would then write rules for
the 3 kind of cells, but only for the rules of birth and death by
isolation. This yields a set of 7 rules; we leave it to the reader to
write them down.

But it may be the case that all cells are regular, which is easy to
achieve for instance by gluing the east and west borders of a
rectangular grid, as well as the north and south ones (there are many
other eccentric ways of ensuring this). For the sake of simplicity we
assume that this is the case of the input graph $G$.

It is easy to see that the set $\Matches{\R}{G}$ is regular: this is
due to the fact that for all matchings $\mm$ only the label of
$\mm(x)$ is modified, and that all matchings with the same $\mm(x)$
perform the same modification (two distinct rules cannot match on the
same center cell of $G$). This aslo means that the result of rewriting
a matching $\mm$ is the same as rewriting all members of the class
$\eqclass{\mm}{\eqaut}$, hence the four rewrite relations defined
above are identical.

However, using parallel rewriting modulo automorphisms may save a lot
of computing efforts compared to full parallel rewriting. For
instance, each time the birth rule matches $G$, with $x$ being matched
to a vertex $v$ of $G$, there must be $3!\times 5! = 720$ matchings from $x$ to
$v$. But they are all equivalent, hence only one of them is required
to compute rewriting modulo automorphisms.

The situation is more complex with the other rules. The rule of death
by overpopulation matches any vertex $v$ of $G$ that has $4\leq n\leq 8$
living cells among its 8 neighbours. The number of matchings from $x$
to $v$ is therefore the number of 4-combinations with repetitions
among $n$ elements, that is, $\frac{n!}{(n-4)!}$. Since the 4
neighbours of $x$ are symmetric in the rule, each matching has an
equivalence class of $4! = 24$ elements, hence the number of
equivalence classes of matchings from $x$ to $v$ is ${n\choose 4}$.
Hence rewriting modulo automorphisms may still require to use as
much as ${8\choose 4} = 70$ different matchings for this rule.

The rule of death by isolation matches any vertex $v$ of $G$ that has
$7\leq m\leq 8$ dead neighbours (or $0\leq n\leq 1$ living neighbours,
since $m=8-n$). The number of matchings from $x$ to $v$ is
$\frac{m!}{(m-7)!}$, and each has an equivalence class of $7!=5040$
elements, hence the number of
equivalence classes of matchings from $x$ to $v$ is $m \choose
7$, which can be as much as ${8\choose 7} = 8$.

\section{Conclusion}\label{sec-conclusion}

 We have defined six properties that deterministic parallel graph rewriting relations should normally meet (though some may be dropped in some contexts), defined four such relations and provided necessary and sufficient conditions for these to obey all six properties simultaneously. In order to obtain these results we have developed a set-theoretic framework for removing and grafting objects in a graph, in which relevant group-theoretic notions can be expressed in a natural way.  We have adopted a notion of graph as general as possible, with the caveat that rewritten graphs must be labelled by ground terms, in order to ensure finiteness of the set of matchings. A more realistic approach would obviously have to encompass other algebras, with ad-hoc restrictions to preserve finiteness, but the issues of regularity and automorphisms would still be central. There are certainly many different ways of ensuring regularity for classes of graphs and rules, and further work is required on this matter. We should also consider a possible implementation of these rewriting relations, which raises the question of using group-theoretic algorithms for efficiently computing the relation of parallel rewriting modulo automorphisms, and in particular of computing a generating set for the automorphism group of a rule.


\end{document}